\documentclass[draftclsnofoot]{IEEEtran}

 \onecolumn
\usepackage{latexsym}
\usepackage{cite}
\usepackage{amssymb}
\usepackage{bm}
\usepackage{amsmath, amssymb}
\usepackage[dvips]{graphics}
\usepackage{graphicx}
\usepackage{psfrag}
\begin{document}
\title{Capacity Region of the Finite-State Multiple Access Channel with and without Feedback}
\author{Haim Permuter and Tsachy Weissman \\
\thanks{The authors are with the Department of Electrical Engineering, Stanford University, Stanford, CA 94305, USA.
(Email: \{haim1, tsachy\}@stanford.edu)}
\thanks{This work was
supported by the NSF through the CAREER award and TFR-0729119
grant.}
}

\markboth{Submitted to IEEE Transactions on Information Theory,
Aug. 2007.}{Shell \MakeLowercase{\textit{et al.}}: Capacity Region
of the Finite-State Multiple Access Channel with and without
Feedback}

\maketitle

\begin{abstract}
The  capacity region of the Finite-State Multiple Access Channel
(FS-MAC)  with feedback that may be an arbitrary time-invariant
function of the channel output samples is considered. We
characterize both an inner and an outer bound for this region, using
Masseys's directed information. These bounds are shown to coincide,
and hence yield the capacity region, of FS-MACs where the state
process is stationary and ergodic and not affected by the inputs.
 Though `multi-letter' in general, our
results  yield explicit conclusions when applied to specific
scenarios  of interest. E.g., our results allow us to:
\begin{itemize}
  \item Identify a large class of FS-MACs, that includes the additive mod-2 noise MAC where the noise may have memory, for which feedback does not enlarge the capacity
  region.

  \item Deduce that, for a general FS-MAC with  states that are not affected by the input, if the capacity (region)  without feedback is zero, then so is
  the capacity (region) with feedback.

\item Deduce that the capacity region of a MAC  that can be decomposed into a
`multiplexer' concatenated by a point-to-point channel (with,
without, or with partial feedback), the capacity region is given
by $\sum_{m} R_m \leq C$, where $C$ is the capacity of the point
to point channel and $m$ indexes the encoders. Moreover, we show
that for this family of channels source-channel coding separation
  holds.
\end{itemize}

\end{abstract}


\newtheorem{question}{Question}
\newtheorem{claim}{Claim}
\newtheorem{guess}{Conjecture}
\newtheorem{definition}{Definition}
\newtheorem{fact}{Fact}
\newtheorem{assumption}{Assumption}
\newtheorem{theorem}{Theorem}
\newtheorem{lemma}[theorem]{Lemma}
\newtheorem{ctheorem}{Corrected Theorem}
\newtheorem{corollary}[theorem]{Corollary}
\newtheorem{proposition}{Proposition}
\newtheorem{example}{Example}
\newtheorem{pfth}{Proof}

\section{Introduction}


The Multiple Access Channel (MAC) has received much attention in the
literature. To put our contributions in context, we begin by briefly
describing some of the key results in the area. The capacity region
for the memoryless MAC
 was derived by Ahlswede in \cite{Ahlswede73MAC}.
 Cover and Leung  derived an achievable
region  for a memoryless MAC with feedback in
\cite{CoverL81MACFeedback}. Using block Markov encoding,
superposition and list codes, they showed that the region $R_1\leq
I(X_1;Y|X_2,U)$, $R_2\leq I(X_2;Y|X_1,U)$ and $R_1+R_2\leq
I(X_1,X_2;Y)$ where
$P(u,x_1,x_2,y)=p(u)p(x_1|u)p(x_2|u)p(y|x_1,x_2)$ is achievable for
a memoryless MAC with feedback.  Willems  showed in \cite{Willems82}
that the achievable region given by Cover and Leung for a memoryless
channel with feedback is optimal for a class of channels where one
of the inputs is a deterministic function of the output and the
other input. More recently Bross and Lapidoth \cite{Bross05Labidoth}
improved  Cover and Leung's region, and Wu et. al.\cite{Wu_Ari06}
have extended Cover and Leung's region for the case that non-causal
state information is available at both encoders.

 Ozarow  derived the capacity of a memoryless Gaussian
MAC  with feedback in \cite{Ozarow84MAC_GaussianFeedback}, and
showed it to be achievable via a modification  of the
Schalkwijk-Kailath scheme
\cite{SchalkwijkKailath66_feedback_scheme}. In general, the
capacity in the presence of noisy feedback is an open question for
the point-to-point channel and a fortiori for the MAC. Lapidoth
and Wigger \cite{LapidothWigger} presented an achievable region
for the case of the Gaussian MAC with noisy feedback and showed
that it converges to Ozarow's noiseless-feedback sum-rate capacity
as the feedback-noise variance tends to zero. Other recent
variations on the Schalkwijk-Kailath scheme of relevance to the
themes of our work include the case of quantization noise in the
feedback link \cite{MartinsWeissman06} and the case of
interference known non-causally at the transmitter
\cite{MerhavWeissman06}.



Verd\'{u} characterized the capacity region of a Multi-Access
channel of the form
$P(y_i|x_1^i,x_2^i,y^{i-1})=P(y_i|x_{1,i-m}^i,x_{2,i-m}^i)$ without
feedback in \cite{Verdu89}. Verd\'{u} further  showed in that work
that in the absence of frame synchronism between the two users,
i.e., there is a random shift between the users, only stationary
input distributions need  be considered. Cheng and Verd\'{u}
 built on the capacity result
from \cite{Verdu89} in \cite{ChengVerdu93_water_filling_Gaussian} to
show that for a Gaussian MAC there exists a water-filling solution
that  generalizes  the point-to-point Gaussian channel.

In \cite{Kramer98}\cite{Kramer03}, Kramer derived several capacity
results for discrete memoryless networks with feedback. By using the
idea of code-trees instead of code-words, Kramer derived a
`mulit-letter' expression for  the capacity of the discrete
memoryless MAC. One of the main results we develop in the present
paper extends Kramer's capacity result to the case of a stationary
and ergodic Markov Finite-State MAC (FS-MAC), to be formally defined
below.

In \cite{Han03}\cite{Han98MAC}, Han  used the information-spectrum
method in order to derive the capacity of a general MAC without
feedback,  when the channel transition probabilities are arbitrary
for every $n$ symbols. Han also considered the additive mod-$q$ MAC,
which we shall use here to  illustrate the way in which our general
results characterize special cases of interest. In particular, our
results will imply that feedback does not increase the capacity
region of the additive mod-$q$ MAC.

In this work, we consider the  capacity region of the Finite-State
Multiple Access Channel (FS-MAC),  with feedback that may be an
arbitrary time-invariant function of the channel output samples. We
characterize  both an inner and an outer bound for this region. We
further show that these bounds coincide, and hence yield the
capacity region, for the important subfamily of FS-MACs with states
that evolve independently of the channel inputs.
 Our derivation of the capacity region is rooted in the derivation
 of the capacity of finite-state channels
   in Gallager's book \cite[ch 4,5]{Gallager68}.
More recently,  Lapidoth and Telatar \cite{Lapidoth98Telatar} have
used it in order to derive the capacity of a compound channel
without feedback, where the compound channel consists of a family of
finite-state channels. In particular, they have introduced into
Gallager's proof the idea of concatenating codewords, which we
extend  here to concatenating code-trees.


 Though `multi-letter' in general, our
results  yield explicit conclusions when applied to more specific
families of MACs. For example, we find that  feedback does not
increase the capacity of the mod-$q$ additive noise MAC (where $q$
is the size of the common alphabet of the input, output and noise),
regardless of the memory in the noise.
 This result is in sharp contrast with the finding of
Gaarder and Wolf in \cite{Wolf75}  that feedback can increase the
capacity even of a  \emph{memoryless} MAC due to cooperation between
senders that it can create. Our result should also be considered in
light of Alajaji's work \cite{Alajaji95},
where it was shown that feedback does not increase the capacity of
discrete point-to-point channels with mod-$q$ additive noise.
Thus, this part of our contribution can be considered a multi-terminal extension of Alajaji's result. 
Our results will in fact allow us to identify a class of MACs larger
than that of  the mod-$q$ additive noise MAC for which feedback does
not enlarge the capacity region.

Further specialization of the results will allow us to deduce
that, for a general FS-MAC with  states that are not affected by
the input, if the capacity (region)  without feedback is zero,
then so is the capacity (region) with feedback. It will also allow
us to identify a large class of FS-MACs for which source-channel
coding separation holds.

The remainder of this paper is organized as follows. We concretely
describe our channel model and assumptions in Section \ref{sec:
Channel Model}. In Section \ref{sec: directed information} we
introduce some notation, tools and results pertaining to directed
information and the notion of causal conditioning that will be key
in later sections. We state our main results in Section \ref{sec:
Main Theorems}. In Section \ref{sec: applications} we apply  the
general results of Section \ref{sec: Main Theorems} to obtain the
capacity region for several interesting classes of channels, as well
as establish a source-channel separation result. The validity of our
inner and outer bounds is established, respectively, in Section
\ref{sec: Proof of Achievability} and Section \ref{sec: Proof of the
outer bound}. In Section \ref{sec: Capacity region of FS-MAC without
feedback} we show that our inner and outer bounds coincide, and
hence yield the capacity region, when applied to the FS-MAC without
feedback. This result can be thought of as the natural extension of
Gallager's results \cite[Ch. 4]{Gallager68} to the MAC or,
alternatively, as the natural extension of Gallager's derivation of
the MAC capacity region in \cite{Gallager85} to channels with
states. In Section \ref{sec: Stationary Finite state Markovian MAC
with feedback} we characterize the capacity region for the case of
arbitrary (time-invariant) feedback and FS-MAC channels with states
that evolve independently of the input, as well as the FS-MAC with
limited ISI (which is the natural MAC-analogue of Kim's
point-to-point channel \cite{Kim07_feedback}),
 by showing that
our inner and outer bounds coincide for this case. We conclude in
Section \ref{sec: conclusions} with a summary of our contribution
and a related future research direction.


\section{Channel Model}
\label{sec: Channel Model} In this paper, we consider an FS-MAC
(Finite state MAC) with a time invariant feedback as illustrated
in Fig. \ref{f_1}.
\begin{figure}[h]{
 \psfrag{A}[][][0.8]{Encoder 1}
\psfrag{B}[][][0.8]{$x_{1,i}(m_1,z_1^{i-1})$}
\psfrag{C}[][][0.8]{Encoder 2}
\psfrag{D}[][][0.8]{$x_{2,i}(m_2,z_2^{i-1})$}
\psfrag{m1}[][][0.8]{$m_1$} \psfrag{m2}[][][0.8]{$\in\{
1,...,2^{nR_1}\}\;\;$} \psfrag{m3}[][][0.8]{$m_2$}
\psfrag{m4}[][][0.8]{$\in\{ 1,...,2^{nR_2}\}\;\;$}
 \psfrag{M}[][][0.8]{Finite State
MAC} \psfrag{P}[][][0.8]{$P(y_i,s_i|x_{1,i},x_{2,i},s_{i-1})$}
\psfrag{f}[][][0.8]{$z_{2,i}=f_2(y_{i})$}
\psfrag{f}[][][0.8]{Time-Invariant} \psfrag{i}[][][0.8]{Function}
\psfrag{zf2}[][][0.8]{$z_{2,i}(y_i)$}
\psfrag{zf1}[][][0.8]{$z_{1,i}(y_i)$}
\psfrag{z2}[][][0.8]{$z_{2,i-1}$}
\psfrag{z1}[][][0.8]{$z_{1,i-1}$} \psfrag{W}[][][0.8]{Decoder}

\psfrag{g}[][][0.8]{Unit} \psfrag{h}[][][0.8]{Delay}
\psfrag{X}[][][0.8]{$\hat m_1(y^N)$} \psfrag{U}[][][0.8]{$\hat
m_2(y^N)$} \psfrag{Y}[][][0.8]{$\hat m_1, \hat m_2$}
\psfrag{D6}[][][0.8]{Function}
\psfrag{v7}[][][0.8]{$z_{i-1}(y_{i-1})$}
 \psfrag{Yi}[][][0.8]{$y_i$}
\psfrag{v6 }[][][0.8]{$\hat m$}\psfrag{w6a}[][][0.8]{Estimated}
\psfrag{w6b\r}[][][0.8]{Message} \centering
\includegraphics[width=13cm]{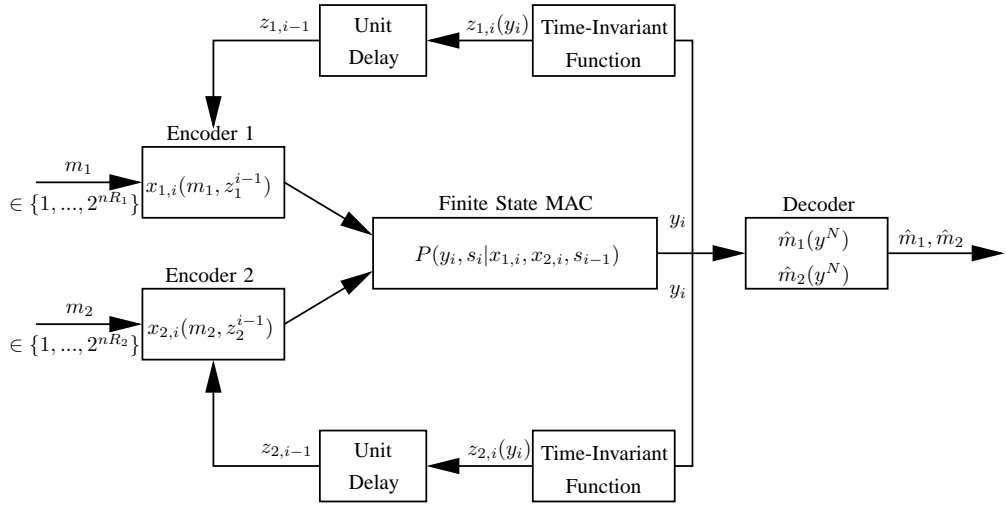}
\centering \caption{Channel with feedback that is a time invariant
deterministic function of the output. } \label{f_1}
}\end{figure} 
The MAC setting consists of two senders and one receiver. Each
sender $l\in\{1,2\}$ chooses an index $m_l$ uniformly from the set
$\{1,...,2^{nR_l}\}$ and independently of the other sender. The
input to the channel from encoder $l$ is denoted by
$\{X_{l1},X_{l2},X_{l3},...\}$, and the output of the channel is
denoted by $\{Y_1,Y_2,Y_3,...\}$. The state at time $i$, i.e.,
$S_i\in \mathcal S$, takes values in a finite set of possible
states. The channel is stationary and is characterized by a
conditional probability $P(y_i,s_i|x_{1i},x_{2i},s_{i-1})$ that
satisfies
\begin{equation}
P(y_i,s_i|x_1^i,x_2^i,s^{i-1},y^{i-1})=P(y_i,s_i|x_{1i},x_{2i},s_{i-1}),
\end{equation}
where the superscripts denote sequences in the following way:
$x_l^i=(x_{l1},x_{l2},...,x_{li}), \; l\in\{1,2\}$. We assume a
communication with feedback $z_l^i$ where the element $z_{li}$ is a
time-invariant function of the output $y_i$. For example, $z_{li}$
could equal $y_i$ (perfect feedback), or a quantized version of
$y_i$, or null (no feedback). The encoders receive the feedback
samples with one unit delay.

A code with feedback consists of two encoding functions
$g_l:\{1,...,2^{nR_1}\}\times \mathcal Z_l^{n-1}\to \mathcal
X_l^n,\; l=1,2$,
where the $k{\text th}$ coordinate of $x_l^n\in\mathcal X_l^n $ is
given by the function
\begin{equation}
x_{lk}=g_{lk}(m_l,z_l^{k-1}), \qquad k=1,2,\ldots ,n, \quad l=1,2
\end{equation}
and a decoding function,
\begin{equation}
g:\mathcal Y^n \to \{1,...,2^{nR_1}\} \times \{1,...,2^{nR_2}\}.
\end{equation}
%
%
The {\it average probability of error} for
$((2^{nR_1},2^{nR_2},n)$ code is defined as
\begin{equation}
P_e^{(n)}=\frac{1}{2^{n(R_1+R_2)}} \sum_{w_1,w_2}
\Pr\{g(Y^n)\neq(w_1,w_2)|(w_1,w_2) \text{ sent}\}.
\end{equation}
A rate $(R_1,R_2)$ is said to be {\it achievable} for the MAC if
there exists a sequence of $((2^{nR_1},2^{nR_2}),n)$ codes with
$P_e^{(n)}\to 0$. The {\it capacity region} of MAC is the closure
of the set of achievebale $(R_1,R_2)$ rates.

%
%
%

\section{Directed Information} \label{sec: directed information}
Throughout this paper we use the {\it Causal Conditioning} notation
$(\cdot||\cdot)$. We denote the probability mass function (pmf) of
$Y^N$ causally conditioned on $X^{N-d}$, for some integer $d\geq0$,
as $P(y^N||x^{N-d})$ which is defined as
\begin{equation} \label{e_causal_cond_def}
P(y^N||x^{N-d})\triangleq \prod_{i=1}^{N} P(y_i|y^{i-1},x^{i-d}),
\end{equation}
 (if $i-d\leq0$ then $x^{i-d}$ is set to null). In particular, we
extensively use the cases where $d=0,1$:
\begin{equation}P(y^N||x^{N})\triangleq \prod_{i=1}^{N}
P(y_i|y^{i-1},x^{i})
\end{equation}
\begin{equation}
Q(x^N||y^{N-1})\triangleq \prod_{i=1}^{N} Q(x_i|x^{i-1},y^{i-1}),
\end{equation}
where the letters $Q$ and $P$ are both used for denoting pmfs.

%
Directed information $I(X^N\to Y^N)$ was defined by Massey in
\cite{Massey90} as
\begin{equation}
I(X^N \to Y^N )\triangleq \sum_{i=1}^{N} I(X^i;Y_i|Y^{i-1}).
\end{equation}
It has been widely used in the characterization of capacity of
point-to-point channels
\cite{Tatikonda00,Chen05,Yang05,Permuter06_feedback_submit,Permuter06_trapdoor_submit,Tatikonda06,Kim07_feedback},
compound channels\cite{ShraderPemuter07ISIT}, network capacity
\cite{Kramer03,Kramer88}, rate distortion
\cite{Pradhan04,Pradhan04b,zamir06} and computational biology
\cite{Hero06,Nuno07}. Directed information can also be expressed
in terms of causal conditioning as
\begin{equation}\label{e_directed_def}
I(X^N \rightarrow Y^N)=\sum_{i=1}^{N} I(X^i;Y_i|Y^{i-1}) = \mathbf
E\left[ \log \frac{P(Y^N||X^N)}{P(Y^N)} \right],
\end{equation}
where $\mathbf E$ denotes expectation. The directed information from
$X^N$ to $Y^N$, conditioned on $S$, is denoted as $I(X^N \to Y^N|S)$
and is defined as:
\begin{equation}
I(X^N \to Y^N|S)\triangleq \sum_{i=1}^{N} I(X^i;Y_i|Y^{i-1},S).
\end{equation}
Directed information between $X_1^N$ to $Y^N$ causally conditioned
on $X_2^N$ is defined as
\begin{equation}
I(X_1^N \to Y^N||X_2^N)\triangleq \sum_{i=1}^{N}
I(X_1^i;Y_i|X_2^i,Y^{i-1}) = \mathbf E\left[ \log
\frac{P(Y^N||X_1^N,X_2^N)}{P(Y^N||X_2^N)} \right].
\end{equation}
where $P(y^N||x_1^N,x_2^N)=\prod_{i=1}^N
P(y_i|y^{i-1},x_1^i,x_2^i)$.

Throughout this paper we are using several properties of causal
conditioning and directed information that follow from the
definitions and simple algebra. Many of the key properties that
hold for mutual information and regular conditioning carry over to
directed information and causal conditioning, where $P(x^N)$ is
replaced by $P(x^N||y^{N-1})$ and
$P(y^N)$ is replaced by $P(y^N||x^N)$. Specifically, 
\begin{lemma} \label{l_joint_causal_condition_decomposition}{(\it Analogue to $P(x_1^N,y^N)=P(x_1^N)P(y^N|x_1^{N})$}.) For arbitrary random
vectors $(X_1^N,X_2^N,Y^N)$,
\begin{equation}
P(x_1^N,y^N)=P(x_1^N||y^{N-1})P(y^N||x_1^{N})
\end{equation}
\begin{equation}
P(x_1^N,y^N||x_2^N)=P(x_1^N||y^{N-1},x_2^N)P(y^N||x_1^{N},x_2^N).
\end{equation}
\end{lemma}
\begin{lemma}\label{l_diff_cond_S}{(\it Analogue to $|I(X_1^N;Y^N)- I(X_1^N;Y^N|S)|\leq H(S)$}.) For
arbitrary random vectors and variables,
\begin{equation}
\left|I(X_1^N \rightarrow Y^N)- I(X_1^N \rightarrow
Y^N|S)\right|\leq H(S)\leq \log|\mathcal{S}|
\end{equation}
\begin{equation}
\left|I(X_1^N \rightarrow Y^N||X_2^{N})- I(X_1^N \rightarrow
Y^N||X_2^{N},S)\right|\leq H(S)\leq
 \log|\mathcal{S}|.
\end{equation}
\end{lemma}
The proofs of Lemma \ref{l_joint_causal_condition_decomposition}
and Lemma \ref{l_diff_cond_S} can be found in \cite[Sec.
IV]{Permuter06_feedback_submit}, along with some additional
properties of causal conditioning and directed information. The
next lemma, which is proven in Appendix
\ref{s_app_lemma_proof_mutual_becomes_directed}, shows that by
replacing regular pmf with causal conditioning pmf we get the
directed information. Let us denote the mutual information
$I(X_1^n;Y^n|X_2^n)$ as a functional of $Q(x_1^N,x_2^N)$ and
$P(y^N|x_1^N,x_2^N)$, i.e., $\mathcal
I(Q(x_1^N,x_2^N);P(y^N|x_1^N,x_2^N))\triangleq
I(X_1^n;Y^n|X_2^n)$. Consider the case that the random variables
$X_1^N,X_2^N$ are independent, i.e.,
$Q(x_1^N,x_2^N)=Q(x_1^N)Q(x_2^N)$, then by definition
\begin{equation}\label{e_calI_def}
\mathcal I(Q(x_1^N)Q(x_2^N);P(y^N|x_1^N,x_2^N))\triangleq
\sum_{y^N,x_1^N,x_2^N}Q(x_1^N)Q(x_2^N)P(y^N|x_1^N,x_2^N)\frac{P(y^N|x_1^N,x_2^N)}{\sum_{{x'}_1^N}Q({x'}_1^N)P(y^N|{x'}_1^N,x_2^N)}.
\end{equation}
\begin{lemma}\label{l_mutual_becomes_directed_causal}
If the random vectors $X_1^N$ and $X_2^N$ are causal-conditionally
independent given $Y^{N-1}$, i.e.,
$Q(x_1^N,x_2^N||y^{N-1})=Q(x_1^N||y^{N-1})Q(x_2^N||y^{N-1})$ then
\begin{equation}
\mathcal
I(Q(x_1^N||y^{N-1})Q(x_2^N||y^{N-1});P(y^N||x_1^N,x_2^N))=I(X_1^N\to
Y^N||X_2^N).
\end{equation}
\end{lemma}
The next lemma, which is proven in Appendix
\ref{s_app_lemma_proof_directed_mutual_no_feedback}, shows that in
the absence of feedback,  mutual information becomes directed
information.
\begin{lemma}\label{l_directed_mutaul_if_no_feedback}
If $Q(x_1^N,x_2^N||y^{N-1})=Q(x_1^N)Q(x_2^N)$ then
\begin{eqnarray}
I(X_1^N;Y^N|X_2^N) = I(X_1^N \to Y^N||X_2^N).
\end{eqnarray}
\end{lemma}

\section{Main Theorems} \label{sec: Main Theorems}

We dedicate this section to a statement of our main results,
proofs of which will appear in the subsequent sections. Let
$\underline {\mathcal R}_n$ denote the following region in
$\mathbb{R}_+^2$ (2D set of nonnegative real numbers):
\begin{equation}\label{e_def_underline_Rn}
\underline {\mathcal
R}_n=\bigcup_{Q(w)Q(x_1^n||z_1^{n-1},w)Q(x_2^n||z_2^{n-1},w)}
\begin{cases}
R_1 \leq  \min_{s_0} \frac{1}{n}I(X_1^n \to Y^n ||X_2^{n},W,s_0)-\frac{\log |\mathcal S|}{n},\\
R_1 \leq  \min_{s_0} \frac{1}{n}I(X_2^n \to Y^n ||X_1^{n},W,s_0)-\frac{\log |\mathcal S|}{n},\\
R_1+R_2 \leq  \min_{s_0} \frac{1}{n}I((X_1,X_2)^n \to
Y^{n}|W,s_0)-\frac{\log |\mathcal S|}{n}.
\end{cases}
\end{equation}
Having the auxiliary random variable $W$ is equivalent to taking the
convex hull of the region. It is shown in the Appendix that the
inclusion (or omission) of $W$ in the definition of the region
$\underline {\mathcal R}_n$ has vanishing effect with increasing
$n$.
\begin{theorem}\label{t_inner_bound}
({\it Inner bound.}) For any FS-MAC with time invariant feedback as
shown in Fig. \ref{f_1}, and for any integer $n\geq1$, the region
$\underline {\mathcal R}_n$ is achievable.
\end{theorem}
Let ${\mathcal R}_n$ denote the following region in
$\mathbb{R}_+^2$
\begin{equation}\label{e_def_Rn}
 {\mathcal R}_n=\bigcup_{Q(x_1^n||z_1^{n-1})Q(x_2^n||z_2^{n-1})}
\begin{cases}
R_1 \leq  \frac{1}{n}I(X_1^n \to Y^n ||X_2^{n}),\\
R_1 \leq  \frac{1}{n}I(X_2^n \to Y^n ||X_1^{n}),\\
R_1+R_2 \leq  \frac{1}{n}I((X_1,X_2)^n \to Y^{n}).
\end{cases}
\end{equation}
In the following theorem we use the standard notion of convergence
of sets. Confer Appendix \ref{s_app_supadditive} for the details
of the definition.
\begin{theorem}\label{t_outer_bound}
({\it Outer bound.}) Let $(R_1,R_2)$ be an achievable pair for a
FS-MAC with time invariant feedback, as shown in Fig. \ref{f_1}.
Then, for any $n$ there exists a distribution
$Q(x_1^n||z_1^{n-1})Q(x_2^n||z_2^{n-1})$ such that the following
inequalities hold:
\begin{eqnarray}
R_1 &\leq&  \frac{1}{n}I(X_1^n \to Y^n ||X_2^{n})+\epsilon_n \nonumber \\
R_2 &\leq&  \frac{1}{n}I(X_2^n \to Y^n ||X_1^{n})+\epsilon_n \nonumber \\
R_1+R_2 &\leq& \frac{1}{n}I((X_1,X_2)^n \to Y^{n})+\epsilon_n,
\end{eqnarray}
where $\epsilon_n$ goes to zero as $n$ goes to infinity. Moreover,
the outer bound can be written as $\liminf {\mathcal R}_n$.
\end{theorem}

For the case where there is no feedback, i.e., $z_i$ is null,
$\mathcal R_n$ and $\underline {\mathcal R}_n$ can be expressed in
terms of mutual information and regular conditioning due to Lemma
\ref{l_directed_mutaul_if_no_feedback}.


\begin{theorem}\label{t_no_feedback}
({\it Capacity of FS-MAC without feedback.}) For any indecomposable
FS-MAC without feedback, the achievable region is
$\lim_{n\to \infty} \mathcal R_n,$
and the limit exists.
\end{theorem}

\begin{theorem}\label{t_capacity_feedback}
({\it Capacity of FS-MAC with feedback.}) For any FS-MAC of the
form
\begin{equation}\label{e_channel_Markov}
P(y_i,s_i|x_{1i},x_{2,i},s_{i-1})=P(s_i|s_{i-1})P(y_i|x_{1i},x_{2,i},s_{i-1}),
\end{equation}
where the state process $S_i$ is stationary and ergodic, the
achievable region is
$\lim_{n\to \infty} \mathcal R_n,$
and the limit exists.
\end{theorem}

The next theorems will be seen to be consequences  of the capacity
theorems given above.
\begin{theorem}  \label{t_capacity_zero} For the channel described in
(\ref{e_channel_Markov}), where the state process $s_i$ is
stationary and ergodic, if the capacity without feedback is zero,
then it is also zero in the case that there is feedback.
\end{theorem}

\begin{corollary}\label{c_memoryless_zero} For a memoryless MAC, the capacity with feedback
is zero if and only if it is zero without feedback.
\end{corollary}

\begin{corollary}\label{c_additive}
Feedback does not enlarge the capacity region of a discrete additive
(mod-$|\mathcal X|$) noise MAC.
\end{corollary}
In fact, among other results, we will see in the next section that
the (mod-$|\mathcal X|$) noise MAC is only a subset of a larger
family of MACs for which feedback does not enlarge the capacity
region.


\section{Applications} \label{sec: applications}
The capacity formula of a FS-MAC given in Theorems
\ref{t_no_feedback} and \ref{t_capacity_feedback} is a
multi-letter characterization. In general, it is very hard to
evaluate it but, for the finite state point to point channel,
there are several cases where the capacity with and without
feedback was found numerically
\cite{Mushkin89}\cite{goldsmith96capacity}, \cite{Yang05},
\cite{Chen05} and analytically
\cite{Permuter06_trapdoor_submit}.\footnote{For the Gaussian case
without feedback there exists the water filling solution
\cite{Shannon48}, and recently the feedback capacity was found
analytically, for the case that the noise is an ARMA(1)-Gaussian
process (cf. \cite{Kim06_MA,yhk06,YangKavcicTatikondaGaussian}).}

The multi-letter capacity expression is also valuable for deriving
useful concepts in communication. For instance, in order to show
that feedback does not increase the capacity of a memoryless
channel (cf. \cite{shannon56}), we can use the multi-letter upper
bound of a channel with memory. Further, in
\cite{Permuter06_feedback_submit} it was shown that for the cases
where the capacity is given by  the multi-letter expression $
 C =\lim_{N \to \infty} \frac{1}{N}  \max_{Q(x^N||z^{N-1})} I(X^N
\rightarrow Y^N )$, the source-channel coding separation holds. It
was also shown that if the state of the channel is known at both
the encoder and decoder and the channel is connected (i.e., every
state can be reached with some positive probability from every
other state under some input distribution), then feedback does not
increase the capacity of the channel.

In this section we use the capacity formula in order to derive
three conclusions:
\begin{enumerate}
\item For a stationary ergodic Markovian channels, the
capacity is zero if and only if the capacity with feedback is zero.
\item Identify FS-MACs that feedback does not enlarge the capacity and show that for a MAC  that can be decomposed into a
`multiplexer' concatenated by a point-to-point channel (with,
without, or with partial feedback), the capacity region is given
by $\sum_{m} R_m \leq C$, where $C$ is the capacity of the point
to point channel.
\item Source-channel coding separation
  holds for a MAC  that can be decomposed into a
`multiplexer' concatenated by a point-to-point channel (with,
without, or with partial feedback).
%
\end{enumerate}

As a special case of the second concept we show that the capacity
of a Binary Gilbert-Ellliot MAC is $R_1+R_2\leq 1-H(\mathcal V)$
where $\mathcal V$ is the entropy rate of the hidden Markov noise
that specifies the Binary Gilbert-Ellliot MAC.

\subsection{Zero capacity}
The first concept is given in Theorem \ref{t_capacity_zero} and is
proved here. The proof of Theorem \ref{t_capacity_zero} is based
on the following lemma which is proven in Appendix
\ref{s_app_proof_of_lemma_zero_iff}.
\begin{lemma} \label{l_zero_iff}
For a MAC described by an arbitrary causal conditioning
$p(y^n||x_1^n,x_2^n)$ the following holds:
\begin{equation}\label{eqn:0iff}
\max_{Q(x_1^n||y^{n-1})Q(x_2^n||y^{n-1})} I(X_1^n,X_2^N\to Y^n)=0
 \iff
\max_{Q(x_1^n)Q(x_2^n)} I(X_1^n,X_2^N\to Y^n)=0,
\end{equation}
and each condition also implies that $P(y^n||x_1^n,x_2^n)=P(y^n)$
for all $x_1^n,x_2^n$.
\end{lemma}

{\it Proof of Theorem \ref{t_capacity_zero}:} Since the channel is
a Markovian channel, i.e.,
\begin{equation}
P(y_i,s_i|x_{1,i},x_{2,i},s_{i-1})=p(s_i|s_{i-1})P(y_i|x_{1,i},x_{2,i},s_{i-1})
\end{equation}
and stationary and ergodic, its capacity region is given in
Theorem \ref{t_capacity_feedback} as $C=\lim_{n\to \infty}
\mathcal R_n$. Furthermore, since the sequence
 $\{\mathcal R_n\}$ is sup-additive (Lemma \ref{l_supadditive_Rn}), then according to Lemma \ref{l_supadditive_region} that is given in Appendix \ref{s_app_supadditive} $\lim_{n\to \infty} \mathcal R_n=\text{cl}\left(\bigcup_{n\geq1}\mathcal
 R_n\right)$, implying that if the capacity without feedback is zero, then for all $n\geq1$
\begin{equation}\label{e_Izero}
\max_{Q(x_1^n)Q(x_2^n)} I(X_1^n,X_2^N\to Y^n)=0.
\end{equation}
According to Lemma \ref{l_zero_iff}, the maximization of the
objective in eq. (\ref{e_Izero}) over the distribution
${Q(x_1^n||y^{n-1})Q(x_2^n||y^{n-1})}$ is still zero,  hence, the
capacity region is zero even if there is perfect feedback. \hfill
\QED

Corollary \ref{c_memoryless_zero}, which states that the capacity
of a memoryless MAC without feedback is zero if and only if the
capacity with feedback is zero, follows immediately from Theorem
\ref{t_capacity_zero} because a memoryless MAC can be considered a
FS-MAC with one state.

Clearly, Theorem \ref{t_capacity_zero} also holds for the case of
a stationary and ergodic FS-Markov point-to-point channel because
a MAC is an extension of a point-to-point channel. However, it
does not hold for the case of a broadcast channel. For instance,
consider the binary broadcast channel given by $y_{1,i}=x\oplus
n_{i}$ and $y_{2,i}=x\oplus n_{i-1}$, where $n_i$ is an i.i.d
Bernoulli$(\frac{1}{2})$ and $\oplus$ denotes addition mod-2. The
capacity without feedback is clearly zero, but if the transmitter
has feedback, namely if it knows $y_{1,i-1}$ and $y_{2,i-1}$ at
time $i$, then it can compute the noise $n_{i-1}=y_{1,i-1}\oplus
x_{i-1}$ and therefore it can transmit 1 bit per channel use to
the second user.

\subsection{Examples of channels for which feedback does not enlarge capacity}

\subsubsection{Gilbert-Elliot MAC}
\begin{figure}[h]{
 \psfrag{P1}[][][0.8]{$\alpha$}
\psfrag{P2}[][][0.8]{$1-\alpha$} \psfrag{P3}[][][0.8]{$\beta$}
\psfrag{P4}[][][0.8]{$1-\beta$} \psfrag{G}[][][1.2]{G}
\psfrag{B}[][][1.2]{B} \psfrag{x1}[][][0.8]{$X_1$}
 \psfrag{x2}[][][0.8]{$X_2$}
 \psfrag{zg}[][][0.8]{$V\sim \text{Bernouli}(p_G)$}
 \psfrag{zb}[][][0.8]{$V\sim \text{Bernouli}(p_B)$}

\centerline{\includegraphics[width=6cm]{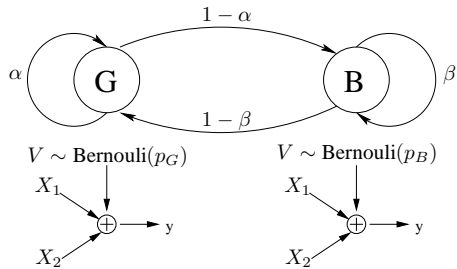}}

\caption{Gilbert-Elliot Mac. It has two states,``Good" and ``Bad"
where the transition between them is according to a first order
Markov process. Given that the channel is in a ``Good" (or a
``Bad") state, it behaves as binary additive noise where the noise
is Bernouli($p_G$) (or Bernouli($p_B$))
\label{f_Gillbert_elliot_mac}}
}\end{figure} 

The Gilbert-Elliot channel is a widely used example of a finite
state channel. It is  often used to model wireless communication
in the presence of fading
\cite{Elliot63,Mushkin89,goldsmith96capacity}. The Gilbert-Elliot
is a Markov channel with two states, denoted as ``good" and
``bad". Each state is a binary symmetric channel and the
probability of flipping the bit is  lower in the ``good" state. In
the case of the Gillber-Elliot MAC (Fig.
\ref{f_Gillbert_elliot_mac}), each state is an additive MAC with
i.i.d noise, where in the ``good" channel the probability that the
noise is '1' is lower than in the bad channel. This channel can be
represented as an additive MAC as in Fig.
\ref{f_Gillbert_elliot_mac}, where the noise is a hidden Markov
process.

Since the Gilbert-Elliot MAC is an ergodic FS-MAC, its capacity
with feedback when the initial state distribution over the states
``good" and ``bad" is the stationary distribution is given by
$\lim_{n\to \infty} \mathcal R_n$ (Theorem
\ref{t_capacity_feedback}). For the Gilbert Elliot MAC, the region
$\lim_{n\to \infty} \mathcal R_n$ reduces to the simple region,
\begin{equation}
R_1+R_2\leq 1-H(\mathcal V) ,
\end{equation}
where $H(\mathcal V)$ denotes the entropy rate of the hidden Markov
noise. The following equalities and inequalities upper bound the
region $\mathcal R_n$ and this upper bound can be achieved for any
deterministic feedback by an i.i.d input distribution $X_{1,i}\sim
\text {Bernoulli}(\frac{1}{2})$ and $X_{2,i}\sim \text
{Bernoulli}(\frac{1}{2}), \; i=1,2,...,n$ and $X_1^n$ and $X_2^n$
are independent of each other.

\begin{eqnarray}\label{e_directed_converse}
I((X_1,X_2)^n\to Y^n) &= &\sum_{i=1}^n H(Y_i | Y^{i-1}) -
H(Y_i | Y^{i-1}, X_1^i,X_2^i) \nonumber \\
&\stackrel{(a)}{=}& \sum_{i=1}^n H(Y_i | Y^{i-1}) -
H(V_i|Y^{i-1},X_1^i,X_2^i) \nonumber \\
&\stackrel{}{=}& \sum_{i=1}^n H(Y_i|Y^{i-1}) - H(V_i | V^{i-1},
Y^{i-1},X_1^i,X_2^i)\nonumber \\
&\stackrel{(b)}{=}& \sum_{i=1}^n H(Y_i | Y^{i-1}) - H(V_i | V^{i-1}) \nonumber \\
&\stackrel{(c)}{\le}& \sum_{i=1}^n \log 2 - H(V_i|V^{i-1}) \nonumber \\
&=& n (1 - \frac{H(V^n)}{n}).
\end{eqnarray}
Equality (a) is due to the facts that $y_i$ is a function of
$(v_i,x_{1,i},x_{2,i})$ and $v_i$ is a deterministic function of
$(y_i,x_{1,i},x_{2,i})$, i.e. $y_i=x_{1,i}\oplus x_{2,i}\oplus v_i$
and $v_i=y_i\oplus x_{1,i}\oplus x_{2,i}$. 
Equality (b) follows from the fact that $v_i$ is independent of
the messages. Inequality (c) is due to the fact that the size of
the alphabet $\cal{Y}$ is $2$. Similarly $\frac{1}{n}I(X_1^n\to
Y^n||X_2^n) \leq 1 - \frac{H(V^n)}{n}$, and $\frac{1}{n}I(X_2^n\to
Y^n||X_1^n) \leq 1 - \frac{H(V^n)}{n}$ and equality is achieved
with an i.i.d input distribution Bernoulli$(\frac{1}{2})$.Finally,
by dividing both sides by $n$ and using the definition of entropy
rate $H(\mathcal V)=\lim_{n \to \infty}\frac{1}{n}H(V^n)$
we conclude the proof. 

%


\subsubsection{Multiplexer followed by a point-to-point channel}
Here we extend the Gilber-Elliot MAC to the case where the
discrete MAC can be decomposed into two components as shown in
Fig. \ref{f_Mac_MUX}. The first component is a MAC that can behave
as a multiplexer and the second component is a point-to-point
channel. The definitions of those components are the following:
\begin{figure}[h]{
 \psfrag{gamma}[][][0.8]{$\gamma$}
\psfrag{W1}[][][0.8]{$W_1$} \psfrag{W2}[][][0.8]{$W_M$}
\psfrag{X1}[][][0.8]{$X_{1i}(W_1,Y^{n-1})$}
\psfrag{Xm}[][][0.8]{$X_{Mi}(W_M,Y^{i-1})$}
\psfrag{X0}[][][0.8]{$X_{0i}$} \psfrag{Z}[][][0.8]{$Z_i$}
\psfrag{Y}[][][0.8]{$Y_i$} \psfrag{Z}[][][0.8]{$Z_i$}
\psfrag{W3}[][][0.8]{$(\hat W_1,\ldots,\hat W_M)$}
\psfrag{X2}[][][0.8]{$X_{1i}(W_1,Y^{i-1})$}
\psfrag{Y2}[][][0.8]{$X_{Mi}(W_M,Y^{i-1})$}
\psfrag{A}[][][2]{$\vdots$} \psfrag{c1}[][][1]{point-to-point}
\psfrag{c2}[][][1]{channel}
 \psfrag{T1}[][][1]{}
\psfrag{T5}[][][1]{} \psfrag{T3}[][][1]{MAC} \psfrag{T4}[][][1]{}
\psfrag{T2}[][][1]{Multiplexer}

\centerline{\includegraphics[width=12cm]{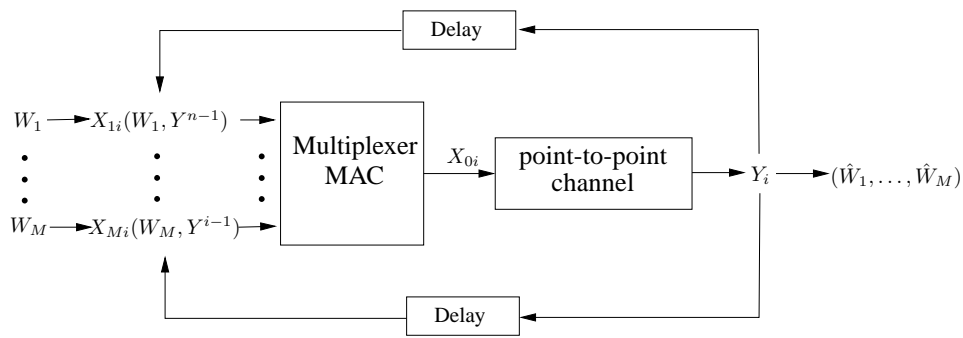}}

\caption{Discrete MAC that can be decomposed into two parts. The
first part is a MAC that behaves as a multiplexer and the second
part is a point-to-point channel \label{f_Mac_MUX}}
}\end{figure} 

\begin{definition} A MAC  behaves as {\it
a multiplexer} if the inputs and the output have common alphabets
and for all $m\in{1,...,M}$ there exists a choice of input symbols
for all senders except sender $m$, such that the output is the
$m\text{th}$ input, i.e. $Y=X_{m}$.
\end{definition}
An example of a multiplexer-MAC for the Binary case is a MAC whose
output is one of and/or/xor of the inputs. For a general alphabet
$q$ those operations could be max/min/addition-mod-$q$. For
instance, if the channel is binary with two users and it is
addition-mod-$2$, i.e., $y=x_1\oplus x_2$, then we can ensure that
$y=x_1$ by choosing $x_2=0$.

\begin{theorem}\label{t_MAC_MUX}
The capacity region of a multiplexer MAC followed by a
point-to-point channel with a time invariant feedback
to all encoders, as shown in Fig. \ref{f_Mac_MUX}, is
\begin{equation}
\sum_{m=1}^M R_m\leq C_{}
\end{equation}
where $C_{}$ is the capacity of the point-to-point channel with the
time invariant feedback $z_{i-1}(y_{i-1})$.
\end{theorem}
\begin{proof}
The achievability is proved simply by time sharing. At each time,
only one selected user sends information and the other users send
a constant input that insures that the output is the input of the
selected user.

The converse is based on the fact that the maximum rate that can
be transmitted through the point-to-point channel is $C$ and it is
an upper bound sum-rate of multiplexer-MAC. If it hadn't been an
upper bound for the multiplexer-MAC, we could build a fictitious
Multiplexer-MAC before the point-to-point channel and achieve by
that a higher rate than its upper bound which would be
contradiction.
\end{proof}
\subsubsection{Discrete additive MAC} An immediate consequence of Theorem \ref{t_MAC_MUX} is
an extension of Alajaj's result \cite{Alajaji95} to the additive
MAC which is given in Corollary \ref{c_additive}. Corollary
\ref{c_additive} states that feedback does not enlarge the
capacity region of a discrete additive (mod-$|\mathcal X|$) noise
MAC.

The proof of the corollary is based on the following observation.
If feedback does not increase the capacity of a particular
point-to-point channel then feedback also does not increase the
capacity of the MUX followed by the same particular channel.
Specifically, feedback does not increase the achievable region of
an additive MAC (Fig. \ref{f_additive}) and the achievable region
is given by
\begin{equation}
\sum_{m=1}^M R_m\leq \log q - H(\mathcal V),
\end{equation}
where $H(\mathcal V)$ is the entropy rate of the additive noise.

\begin{figure}[h]{
 \psfrag{gamma}[][][0.8]{$\gamma$}
\psfrag{W1}[][][0.8]{$W_1$} \psfrag{W2}[][][0.8]{$W_M$}
\psfrag{X}[][][0.8]{$X_{1n}(W_1)$}
\psfrag{Y}[][][0.8]{$X_{Mn}(W_M)$} \psfrag{Z}[][][0.8]{$Y_{n}$}
\psfrag{N}[][][0.8]{$V_n$} \psfrag{W3}[][][0.8]{$(\hat
W_1,\ldots,\hat W_M)$} \psfrag{X2}[][][0.8]{$X_{1n}(W_1,Y^{n-1})$}
\psfrag{Y2}[][][0.8]{$X_{Mn}(W_M,Y^{n-1})$}
\psfrag{A}[][][2]{$\vdots$}

\centerline{\includegraphics[width=16cm]{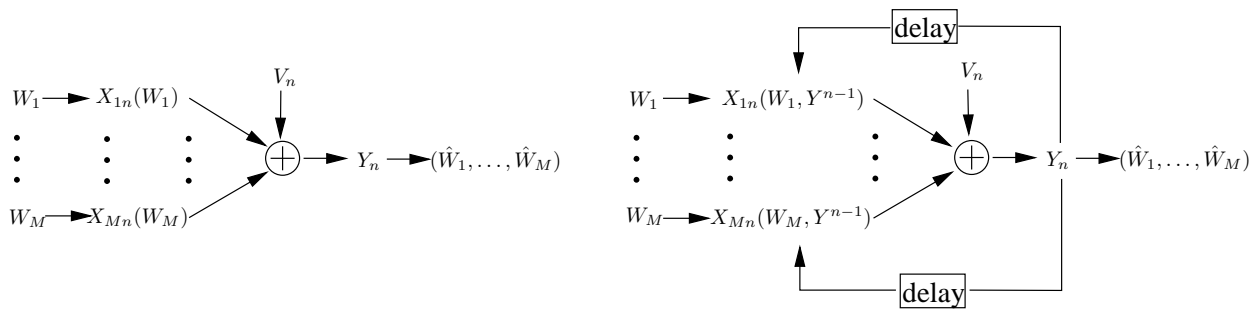}}

\caption{Additive noise MAC with and without feedback. The random
variables $X_{1n},...,X_{Mn},Y_n,V_n, \quad n\in{1,2,3,...}$, are
from a common alphabet of size $q$, and they denote the input from
sender 1,...,M, the output and the noise at time $n$,
respectively. The relation between the random variables is given
by $y_n=x_{1n}\oplus x_{2n} ... \oplus x_{Mn} \oplus v_{n}$ where
$\oplus$ denotes addition mod-$q$. The noise $V_n$, possibly with
memory, is independent of the messages
$W_{1},...,W_{M}$.\label{f_additive} }
}\end{figure} 

\subsubsection{Multiplexer followed by erasure channel} Consider the case of the multiplexer-erasure MAC which is a
multiplexer followed by an erasure channel, possibly with memory.
\begin{definition}
A point-to-point channel is called {\it erasure} channel if the
output at time $n$ can be written as $Y_n=f(X_n,Z_n)$, and the
following properties hold:
\begin{enumerate}
\item The alphabet of $Z$ is binary and the alphabet of $Y$ is the same as $X$ plus one additional symbol called
the erasure.
\item The process $Z_n$ is stationary and ergodic and is independent of the message.
\item If $z_n=0$, then $y_n=x_n$ and if $z_n=1$, then the output is an erasure regardless of the input.
\end{enumerate}
\end{definition}
 For the
mutltiplexr-erasure channel we have the following theorem.
\begin{corollary}\label{c_multiplexer-erasure}
The capacity region of the multiplexer-erasure MAC with or without
feedback is
\begin{equation}\label{e_capacity_erasure}
\sum_{m=1}^M R_m \leq (1 - p_e)\log q,
\end{equation}
where $p_e$ is the marginal probability of having an erasure.
Moreover, even if the encoder has non causal side information,
i.e. the encoders know where the erasures appear noncausally, the
capacity is still given by (\ref{e_capacity_erasure}).
\end{corollary}
%
%
%
%
%

{\it Proof:} According to Theorem \ref{t_MAC_MUX} the capacity
region is
\begin{equation}
\sum_{m=1}^M R_m \leq C,
\end{equation}
where $C$ is the capacity of the erasure point-to-point channel.
Diggavi and Grossglauser \cite[Thm.
3.1]{Diggavi06_deleteion_erasur_channels} showed that the capacity
of a point-to-point erasure channel, with and without feedback, is
given by $(1 - p_e)\log q$. Since the probability of having an
erasure does not depend on the input to the channel, we deduce
%
%
%
that even in the the case where the encoder knows the sequence
$Z^n$ non-causally, which is better than feedback, the transmitter
can transmit only fraction $1-p_e$ of the time, hence the capacity
cannot exceed $ (1-p_e)\log q$.
\hfill \QED

\subsubsection{Multiplexer followed by the trapdoor channel}
In this example feedback increases the capacity. Based on the fact
that the capacity of the trapdoor channel with
feedback\cite{Permuter06_trapdoor_submit} is the logarithm of the
golden ratio, i.e. $\log\frac{\sqrt{5}+1}{2}$, the achievable
region of a Multiplexer followed by the trapdoor channel is
\begin{equation}
\sum_{m=1}^M R_m \leq \log\frac{\sqrt{5}+1}{2}.
\end{equation}

\subsection{Source-channel coding separation}
Cover, El-Gamal and Salehi \cite{CoverGamalSalehi80MAC_Source}
showed that, in general, the source channel separation does not
hold for MACs even for a memoryless channel without feedback.
However, for the case where the MAC is a discrete Multiplexer
followed by a channel we now show that it does hold.

We want to send the sequence of symbols $U_1^n,U_2^n$ over the
MAC, so that the receiver can reconstruct the sequence. To do this
we can use a joint source-channel coding scheme where we send
through the channel the symbols $x_{1,i}(u_1^n,z^{i-1})$ and
$x_{2,i}(u_2^n,z^{i-1})$. The receiver looks at his received
sequence $Y^n$ and makes an estimate $\hat U_1^n, \hat U_2^n$. The
receiver makes an error if $\hat U_1^n\neq U_1^n$ or if $\hat
U_2^n\neq U_2^n$, i.e.,  the probability of error $P_e^{(n)}$ is
$P_e^{(n)}=\Pr( (\hat U_1^n,\hat U_2^n)\neq ( U_1^n, U_2^n)).$
\begin{theorem} {(\it Source-channel coding theorem for a Multiplexer followed by a channel.})
 Let $(U_1,U_2)_{n\geq1}$ be a finite alphabet, jointly stationary and ergodic  pair of 
 processes and let the MAC channel be a multiplexer followed by a point-to-point channel with  time invariant
 feedback and capacity $C=\lim_{N\to\infty}\frac{1}{N}\max_{Q(x^n||z^{n-1})} I(X^n;Y^n)$ (e.g., a memoryless channel, an indecomposable FSC without feedback, stationary and ergodic Markovian channel).
 For the source and the MAC described above:

{(\it direct part.)} There exists
 a source-channel code  with $P_e^{(n)}\to 0$, if $H(\mathcal U_1,\mathcal
 U_2)< C$, where $H(\mathcal U_1,\mathcal
 U_2)$ is the entropy rate of the sources and $C$ is the capacity
 of the point-to-point channel with a time-invariant feedback.

{(converse part).} If $H(\mathcal U_1,\mathcal
 U_2)>C$, then the probability of error is bounded away from zero (independent of the blocklength). 
\end{theorem}
\begin{proof} The achievability is  a straightforward
consequence of the Slepian-Wolf result for Ergodic and stationary
processes \cite{Cover75SliepenWolf} and the achievability of the
multiplexer followed by a point-to-point channel. First, we encode
the sources by using the Sepian-Wolf achievability scheme where we
assign every $u_1^n$ to one of $2^{nR_1}$ bins according to a
uniform distribution on $\{1,...,2^{nR_1}\}$ and independently we
assign every $u_2^n$ to one of $2^{nR_2}$ bins according to a
uniform distribution on $\{1,...,2^{nR_2}\}$. Second, we encode
the bins as if they were messages, as shown in Fig.
\ref{f_Mac_MUX_seperation}.

In the converse, we assume that there exists a sequence of codes
with  $P_e^{(n)}\to 0$,  and we show that it implies that
$H(\mathcal U_1,\mathcal U_2)\leq C$. Fix a given coding scheme
and consider the following:

 \begin{eqnarray}\label{e_HU_1U_2}
H(U_1^n,U_2^n)&\stackrel{(a)}{\leq}& I(U_1^n,U_2^n;\hat U_1^n,
\hat
U_2^n) + n\epsilon_n\nonumber \\
&\stackrel{(b)}{\leq}& I(U_1^n,U_2^n;Y^n) + n\epsilon_n\nonumber \\
&\stackrel{=}{}& H(Y^n)-H(Y^n|U_1^n,U_2^n) + n\epsilon_n\nonumber \\
&\stackrel{=}{}& \sum_{i=1}^n H(Y_i|Y^{i-1})-H(Y_i|U_1^n,U_2^n,Y^{i-1}) + n\epsilon_n\nonumber \\
&\stackrel{(c)}{=}& \sum_{i=1}^n H(Y_i|Y^{i-1})-H(Y_i|U_1^n,U_2^n,Y^{i-1},X_1^i,X_2^i) + n\epsilon_n\nonumber \\
&\stackrel{(d)}{=}& \sum_{i=1}^n H(Y_i|Y^{i-1})-H(Y_i|Y^{i-1},X_1^i,X_2^i) + n\epsilon_n\nonumber \\
&\stackrel{=}{}& \sum_{i=1}^n H(Y_i|Y^{i-1})-H(Y_i|Y^{i-1},X_1^i,X_2^i) + n\epsilon_n\nonumber \\
&\stackrel{=}{}& \sum_{i=1}^n I(X_1^i,X_2^i;Y_i|Y^{i-1}) + n\epsilon_n\nonumber \\
&\stackrel{(e)}{\leq}& \sum_{i=1}^n I(X_{0}^i;Y_i|Y^{i-1}) + n\epsilon_n\nonumber \\
&\stackrel{}{=}& I(X_{0}^n \to Y^n) + n\epsilon_n\nonumber \\
&\stackrel{}{\leq}& \max_{Q(x_0^n||z^{n-1})}I(X_{0}^n \to Y^n) + n\epsilon_n 
\end{eqnarray}
Inequality (a) is due to Fano's inequality where $n
\epsilon_n=1+P_e^{(n)}n|\mathcal U_1||\mathcal U_2|$. Inequality
(b) follows from the data processing inequality because
$(U_1^N,U_2^N)-Y^N-(\hat U_1^N,\hat U_2^N)$ form a Markov chain.
Equality (c) is due to the fact that, for a given code, $X_1^i$ is
a deterministic function of $U_1^n,Y^{i-1}$ and, similarly,
$X_2^i$ is a deterministic function of $U_2^n,Y^{i-1}$. Equality
(d) is due to the Markov chain
$(U_1^N,U_2^N)-(X_1^i,X_2^i,Y^{i-1})-Y_i$. The notation $X_{0,i}$
denotes the output of the multiplexer which is also the input to
the point-to-point channel at time $i$. The inequality in $(e)$ is
due to the data processing inequality which can be invoked thank
to the fact that given $Y^{i-1}$ we have the Markov chain
$X_{1}^i,X_{2}^i-X_{0}^i-Y_i$.

By dividing both sides of (\ref{e_HU_1U_2}) by $n$, taking the
limit $n\to \infty$, and recalling that
$C=\lim_{n\to\infty}\frac{1}{n}\max_{Q(x^n||z^{n-1})} I(X^n;Y^n)$
we have
\begin{eqnarray}
H(\mathcal U_1,\mathcal U_2)=\lim_{n\to\infty}\frac{1}{n}
H(U_1^n,U_2^n)\leq C.
\end{eqnarray}
%
%
\end{proof}
\begin{figure}[h]{
\psfrag{W1}[][][0.8]{$W_1(U_1^n)$} \psfrag{U1}[][][0.8]{$U_1^n$}
\psfrag{s1}[][][0.7]{$\in\{1,...,2^{nR_1}\}$}

\psfrag{W2}[][][0.8]{$W_2(U_2^N)$}\psfrag{U2}[][][0.8]{$U_2^n$}
\psfrag{s2}[][][0.7]{$\in\{1,...,2^{nR_2}\}$}

\psfrag{X1}[][][0.8]{$X_{1i}(W_1,Y^{n-1})$}
\psfrag{Xm}[][][0.8]{$X_{2i}(W_2,Y^{i-1})$}
\psfrag{X0}[][][0.8]{$X_{0i}$} \psfrag{Z}[][][0.8]{$Z_i$}
\psfrag{Y}[][][0.8]{$Y_i$} \psfrag{Z}[][][0.8]{$Z_i$}
\psfrag{W3a}[][][0.8]{$\hat W_1(Y^n)$} \psfrag{W3b}[][][0.8]{$\hat
W_2(Y^n)$}

\psfrag{U3a}[][][0.8]{$\hat U_1^n(\hat W_1,\hat W_2)$}
\psfrag{U3b}[][][0.8]{$\hat U_2^n(\hat W_1,\hat W_2)$}

\psfrag{X2}[][][0.8]{$X_{1i}(W_1,Y^{i-1})$}
\psfrag{Y2}[][][0.8]{$X_{2i}(W_2,Y^{i-1})$}
\psfrag{A}[][][2]{$\vdots$} \psfrag{c1}[][][1]{point-to-point}
\psfrag{c2}[][][1]{channel}
 \psfrag{T1}[][][1]{}
\psfrag{T5}[][][1]{} \psfrag{T3}[][][1]{MAC} \psfrag{T4}[][][1]{}
\psfrag{T2}[][][1]{Multiplexer}

\centerline{\includegraphics[width=15cm]{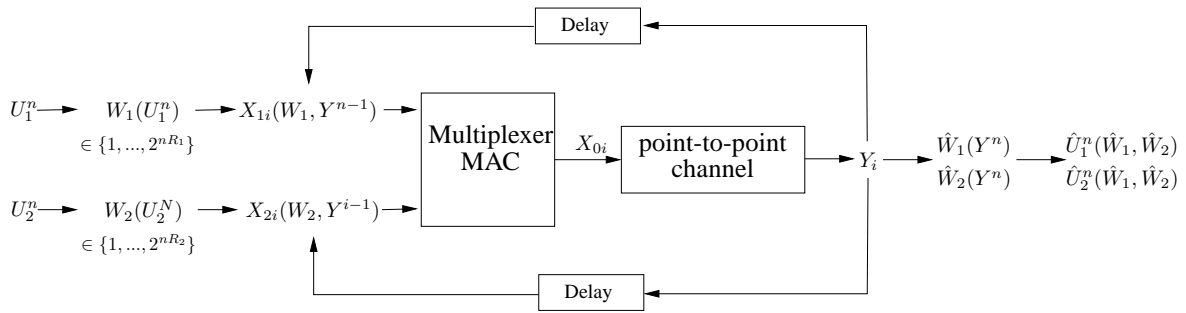}}

\caption{Source-channel coding separation in a discrete Multiplexer
followed by a point-to-point channel. \label{f_Mac_MUX_seperation}}
}\end{figure} 

\section{Proof of Achievability (Theorem \ref{t_inner_bound})} \label{sec: Proof of Achievability}
The proof of achievability for the FS-MAC with feedback is similar
to the proof of achievability for the point-to-point FSC given in
\cite[Sec. V]{Permuter06_feedback_submit}, but there are two main
differences:
\begin{enumerate}
\item In the case of FSC, only one message is sent, and in the case of FS-MAC, two independent messages are
sent, which requires that we analyze three different types of
errors: the first type occurs when only the first message is decoded
with error, the second type occurs when only the second message is
decoded with error, and the third type occurs when both messages are
decoded with error.
\item In both cases, we generate the encoding scheme (code-trees) randomly but the distribution that is used is different. In
the case of FSC we generate, for each message in $[1,...,2^{NR}],$
a code-tree of length $N$ by using the causal conditioning
distribution
$Q^*(x^N||z^{N-1})=\arg\max_{Q(x^N||z^{N-1})}\min_{s_0}I(X^N\to
Y^N|s_0)$, and here we generate for each message in
$[1,...,2^{NR_l}], l=1,2$ a code-tree of length $N=Kn$ by
concatenating $K$ independent code-trees where each one is created
with a causal conditioning distribution $Q(x_l^n||z_l^{n-1}),
l=1,2$. 
\end{enumerate}

{\bf Encoding scheme:} Randomly generate for encoder
$\;\{l\in{1,2}\}$, $\;2^{NR_l}$ code-trees of length $N=Kn$ by
drawing it with the fixed distributions $Q(x_l^n||z_l^{n-1})$. In
other words, given a feedback sequence $z_1^{N-1}$ the causal
conditioning probability that the sequence $x_1^N$ will be mapped to
a given message is
\begin{equation} \label{e_conacatinationK}
Q(x_1^N||z_1^{N-1})=\prod_{k=1}^K Q(x_{1, (k-1)n+1}^{kn}||z_{1,
(k-1)n+1}^{kn}),
\end{equation} where $x_{1,
(k-1)n+1}^{kn}$  denotes the vector $(x_{1, (k-1)n+1},x_{1,
(k-1)n+2},...,x_{1, kn} )$. Fig. \ref{f_codetree} illustrates the
concatenation of trees graphically. In order to shorten the
notation we will sometimes use the notation $Q_N$ to denote
$Q(x_1^N||z_1^{N-1})Q(x_2^N||z_2^{N-1})$  and we will express the
concatenation of pmfs in (\ref{e_conacatinationK}) as
$Q_N=\prod_{k=1}^K Q_n$.
%
%

\begin{figure}[h]{
\psfrag{c1}[][][0.8]{$x_1=0$} \psfrag{a1}[][][0.8]{$$}
\psfrag{c2}[][][0.8]{$x_2=1$} \psfrag{a2}[][][0.8]{$i=1$}
\psfrag{c3}[][][0.8]{$x_3=1$} \psfrag{a3}[][][0.8]{$i=2$}
\psfrag{c4}[][][0.8]{$x_4=0$} \psfrag{a4}[][][0.8]{$i=3$}

\psfrag{d1}[][][0.8]{$x_1=0$} \psfrag{a1}[][][0.8]{$$}
\psfrag{d2}[][][0.8]{$x_2=1$} \psfrag{a1}[][][0.8]{$$}
\psfrag{d3}[][][0.8]{$x_2=1$} \psfrag{a1}[][][0.8]{$$}
\psfrag{d4}[][][0.8]{$x_3=0$} \psfrag{a1}[][][0.8]{$$}
\psfrag{d5}[][][0.8]{$x_3=1$} \psfrag{a1}[][][0.8]{$$}
\psfrag{d6}[][][0.8]{$x_3=1$} \psfrag{a1}[][][0.8]{$$}
\psfrag{d7}[][][0.8]{$x_3=1$} \psfrag{a1}[][][0.8]{$$}
\psfrag{da}[][][0.8]{$x_4=0$} \psfrag{db}[][][0.8]{$x_4=1$}
\psfrag{a1}[][][0.8]{$$} \psfrag{a5}[][][0.8]{$i=4$}

\psfrag{k0}[][][0.8]{$\;\;\;\;\;\;z_{i-1}=0$}
\psfrag{k1}[][][0.8]{$\;\;\;\;\;\;z_{i-1}=1$}
\psfrag{k01}[][][0.8]{$\:\:\:\:\:\:\:\:\:\:\:\:$ (no feedback)}

\psfrag{codeword}[][][0.9]{codeword (case of no feedback)}
\psfrag{code-tree}[][][0.9]{code-tree (used in
\cite{Permuter06_feedback_submit})} \psfrag{con
code-tree}[][][0.9]{concatenated code-tree (used here)}

\centerline{\includegraphics[width=16cm]{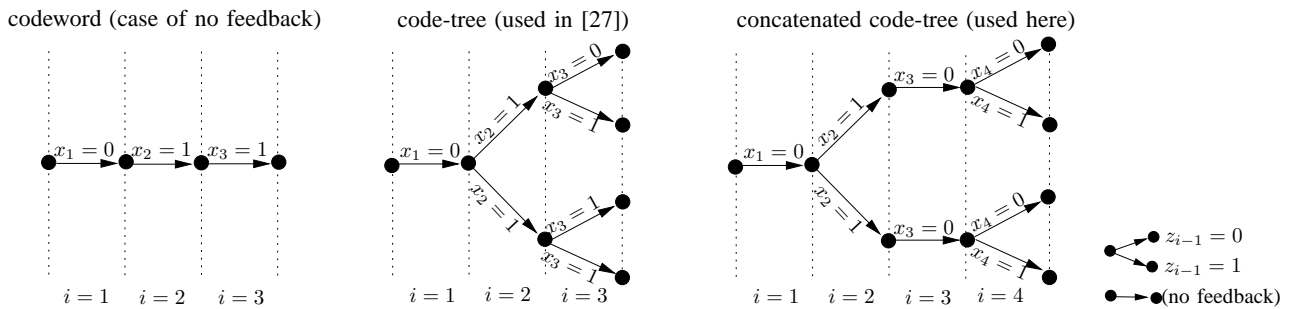}}
\caption{Illustration of coding scheme for setting without feedback,
setting with feedback as used for point-to-point channel
\cite{Permuter06_feedback_submit} and a code-tree that was created
by concatenating smaller code-trees. In the case of no feedback each
message is mapped to a codeword, and in the case of feedback each
message is mapped to a code-tree. The third scheme is a code-tree of
depth 4 created by concatenating two trees of depth 2. }
\label{f_codetree} }
\end{figure}

{\bf Decoding Errors:} For each code in the ensemble, the decoder
uses maximum likelihood decoding and we want to upper bound the
expected value ${\bf E} [P_e]$ for this ensemble. Let
$P_{e1},P_{e2},P_{e3}$ be defined as follows.
\begin{itemize}
\item[$P_{e1}$] (type $1$ error): probability that the decoded pair $(m_1,m_2)$ satisfies $\hat m_1 \neq
m_1, \hat m_2 = m_2$,

\item[$P_{e2}$] (type $2$ error): probability that the decoded pair $(m_1,m_2)$ satisfies $\hat m_1=
m_1, \hat m_2 \neq m_2$,

\item[$P_{e3}$] (type $3$ error): probability that the decoded pair $(m_1,m_2)$ satisfies $\hat m_1 \neq m_1, \hat m_2 \neq
m_2$.

\end{itemize}
Because the error events are disjoint we have
\begin{equation}
P_e=P_{e1}+P_{e2}+P_{e3}
\end{equation}
In the next sequence of theorems and lemmas, we upper bound the
expected value of each error type and show that if $(R_1,R_2)$
satisfies the three inequalities that define $\underline{\mathcal
R}_n$ then the corresponding ${\bf E}[P_{ei}], i=1,2,3$ goes to zero
and hence ${\bf E} [P_e]$ goes to zero.

\begin{theorem}
\label{t_MLB} Suppose that an arbitrary message $m_1,m_2, 1\leq m_1
\leq M_1, 1\leq m_2 \leq M_2$, enters the encoder with feedback and
that ML decoding is employed. Let $E[P_{e_1}|m_1,m_2]$ denote the
probability of decoding error averaged over the ensemble of codes
when the messages $m_1,m_2$ were sent. Then for any choice of
$\rho,0 < \rho \leq 1$,
\begin{eqnarray}
\mathbf E[P_{e_1}|m_1,m_2] &\leq& (M_1-1)^{\rho}
\sum_{y^N,x_2^N}Q(x_2^N||z^{N-1}) \left[ \sum_{x_1^N}
Q(x_1^N||z_1^{N-1}) P(y^N||x_1^N,x_2^N)^{\frac{1}{(1+\rho)}} \right]
^{1+\rho},\label{e_MLB_a}\\
\mathbf E[P_{e_2}|m_1,m_2] &\leq& (M_2-1)^{\rho}
\sum_{y^N,x_1^N}Q(x_1^N||z_1^{N-1}) \left[ \sum_{x_2^N}
Q(x_2^N||z^{N-1}) P(y^N||x_1^N,x_2^N)^{\frac{1}{(1+\rho)}} \right]
^{1+\rho},\label{e_MLB_b}\\
\mathbf E[P_{e_3}|m_1,m_2] &\leq& ((M_2-1)(M_2-1))^{\rho}
\sum_{y^N}\left[ \sum_{x_1^N,x_2^N}
Q(x_1^N||z_1^{N-1})Q(x_2^N||z^{N-1})
P(y^N||x_1^N,x_2^N)^{\frac{1}{(1+\rho)}} \right] ^{1+\rho}. \nonumber \\
& & \label{e_MLB_c}
\end{eqnarray}
\end{theorem}
The proof is given in Appendix \ref{s_app_proof_MLB} and is similar
to \cite[Theorem 9]{Permuter06_feedback_submit} only that here we
take into account the fact that there are two encoders rather than
one.

Let $P_{ei}(s_0), i=1,2,3$ be the probability of error of type $i$
given that the initial state of the channel is $s_0$. Also let
$R_1=\frac{1}{N}\log M_1$ and $R_2=\frac{1}{N}\log M_2$ be the rate
of the code and $R_3$ be the sum rate, i.e. $R_3=R_1+R_2$. The
following theorem establishes exponential  bounds on  ${\bf E}
[P_{ei}(s_0)]$.

\begin{theorem} \label{t_fn}
The average probability of error over the ensemble, for all initial
states $s_0$, and all $\rho$, $0 \leq \rho \leq 1$, is bounded as
\begin{equation}\label{e_pfn}
{\bf E}[P_{ei}(s_0)|m_1,m_2]\leq|\mathcal{S}|2^{\{-N[-\rho
R_i+F_{N,i}(\rho, Q_N)]\}},\qquad i=1,2,3
\end{equation}
where
\begin{equation} \label{eq_fn}
F_{N,i}(\rho, Q_N)= -\frac{\rho \log |\mathcal{S}|}{N}+\left[
\min_{s_0} E_{N,i}(\rho, Q_N,s_0)\right],\qquad i=1,2,3
\end{equation}
\begin{eqnarray}\label{eq_E}
E_{N,1}(\rho,Q_N,s_0)&=&-\frac{1}{N}\log
\sum_{y^N,x_2^N}Q(x_2^N||z^{N-1}) \left[ \sum_{x_1^N}
Q(x_1^N||z_1^{N-1}) P(y^N||x_1^N,x_2^N,s_0)^{\frac{1}{(1+\rho)}}
\right]
^{1+\rho}\\
E_{N,2}(\rho,Q_N,s_0)&=&-\frac{1}{N}\log\sum_{y^N,x_1^N}Q(x_1^N||z_1^{N-1})
\left[ \sum_{x_2^N} Q(x_2^N||z^{N-1})
P(y^N||x_1^N,x_2^N.s_0)^{\frac{1}{(1+\rho)}} \right]
^{1+\rho}\\
E_{N,3}(\rho,Q_N,s_0)&=&-\frac{1}{N}\log \sum_{y^N}\left[
\sum_{x_1^N,x_2^N} Q(x_1^N||z_1^{N-1})Q(x_2^N||z^{N-1})
P(y^N||x_1^N,x_2^N,s_0)^{\frac{1}{(1+\rho)}} \right] ^{1+\rho}.
\label{eq_E3}
\end{eqnarray}
\end{theorem}
The proof is based on algebraic manipulation of the bounds given in
(\ref{e_MLB_a})-(\ref{e_MLB_c}). It is similar to the proof of
Theorem 9 in \cite{Permuter06_feedback_submit} and therefore
omitted. There are two differences between the proofs (and both are
straightforward to accommodate): Here the input distribution
$Q_N=Q(x_1^N||z_1^N)Q(x_2^N||z_2^N)$ is arbitrary while in
\cite{Permuter06_feedback_submit} we chose the one that maximizes
the error exponent. Second, here we bound the averaged error over
the ensemble and in \cite{Permuter06_feedback_submit} we have an
additional step where we claim that there exists a code that has an
error that is bounded by the expression in (\ref{e_pfn}). Because of
this difference the bound on the probability of error in
\cite{Permuter06_feedback_submit} has an additional factor of $4$.

The following theorem presents a few properties of the functions
$E_{N,i}(\rho,Q_N,s_0), \; i=1,2,3$, such as positivity of the
function and its derivative, convexity with respect to $\rho$, and
an upper bound on the derivative which is achieved for $\rho=0$.
\begin{lemma} \label{l_der}
The term $E_{N,i}(\rho,Q_N,s_0)$ has the following properties:

\begin{equation}\label{eq_e1}
E_{N,i}(\rho,Q_N,s_0) \geq 0; \quad \rho \geq 0, i=1,2,3,
\end{equation}

\begin{eqnarray}\label{eq_e2}
 \frac{1}{N}I(X_1^N\to Y^N||X_2^N,s_0) &\geq&\frac{\partial
E_{N,1}(\rho, Q_N,s_0)}{\partial \rho} > 0; \quad \rho \geq 0\nonumber \\
 \frac{1}{N}I(X_2^N\to Y^N||X_1^N,s_0) &\geq&\frac{\partial
E_{N,2}(\rho, Q_N,s_0)}{\partial \rho} > 0; \quad \rho \geq 0\nonumber \\
 \frac{1}{N}I(X_1^N,X_2^N\to Y^N|s_0) &\geq&\frac{\partial
E_{N,3}(\rho, Q_N,s_0)}{\partial \rho} > 0; \quad \rho \geq 0
\end{eqnarray}

\begin{equation}
\frac{\partial^2 E_{N,i}(\rho,Q_N,s_0)}{\partial \rho^2}
> 0; \quad \rho \geq 0, i=1,2,3.
\end{equation}
Furthermore, equality holds in (\ref{eq_e1}) when $\rho=0$, and
equality holds on the left sides of eq. (\ref{eq_e2}) when $\rho=0$
for $i=1,2,3$.\end{lemma}

The proof of the theorem is the same proof as  \cite[eq.
(2.20)]{Gallager85}, \cite[Theorem 5.6.3]{Gallager68}. In
\cite{Gallager85} the arguments $Q_N$ of $E_{N,1}(\rho, Q_N,s_0)$
are regular conditioning i.e., $Q(x_1^N)Q(x_2^N)$, and the channel
is given by $P(y^N|x_1^N,x_2^N,s_0)$, hence the derivative of
$E_{N,1}(\rho, Q_N,s_0)$ with respect to $\rho$ is upper-bounded by
$I(X_1^N;Y^N|X_2^N,s_0)$. Here we replace $Q(x_1^N)Q(x_2^N)$ with
$Q(x_1^N||z_1^{N-1})Q(x_2^N||z_2^{N-1})$ and
$P(y^N|x_1^N,x_2^N,s_0)$ with $P(y^N||x_1^N,x_2^N,s_0)$ and,
according to Lemma \ref{l_mutual_becomes_directed_causal}, the
upper-bound becomes $I(X_1^N\to Y^N||X_2^N,s_0)$.
%
The next lemma establishes the sup-additivity of $F_{N,i}(\rho,Q_N),
i=1,2,3$.
\begin{lemma} \label{l_fn}
{\it Sup-additivity of $F_{N,i}(\rho,Q_N)$.} For any finite-state
channel, $F_{N,i}(\rho,Q_N)$, as given by eq. (\ref{eq_fn}),
satisfies
\begin{equation}
F_{n+l,i}(\rho,Q_{n+l})\geq \frac{n}{n+l} F_{n,i}(\rho,Q_n) +
\frac{l}{n+l}F_{l,i}(\rho,Q_l), \qquad i=1,2,3.
\end{equation}
\end{lemma}
The proof steps are identical to the proof of the sub-additivity for
the point-to-point channel \cite[Lemma
11]{Permuter06_feedback_submit}.

Invoking this lemma on the pmf $Q_N=\prod_{k=1}^K Q_n$ where $N=nK$
 we get
\begin{equation}\label{e_FN_n}
 F_{N,i}(\rho,Q_N)\geq
K \frac{n}{N} F_{n,i}(\rho,Q_n)=F_{n,i}(\rho,Q_n).
\end{equation}

Let us define
 \begin{eqnarray}
 \underline C_{N,1}(Q_N) &=& \frac{1}{N}  \min_{s_0}I(X_1^N
\rightarrow Y^N ||X_2^N, s_0)\\
 \underline C_{N,2}(Q_N) &=& \frac{1}{N}  \min_{s_0}I(X_2^N
\rightarrow Y^N ||X_1^N, s_0)\\
 \underline C_{N,3}(Q_N) &=& \frac{1}{N}  \min_{s_0}I(X_1^N,X_2^N
\rightarrow Y^N | s_0)
 \end{eqnarray}
 where the joint distribution of $X_1^N,X_2^N,Y^N$ conditioned on $s_0$ is given by
 $P(x_1^N,x_2^N,y^N|s_0)=Q(x_1^N||z_1^{N-1})Q(x_2^N||z_2^{N-1})P(y^N||x_1^N,x_2^N,s_0)$.

%
Theorem \ref{t_inner_bound} (inner bound) given in Sec. \ref{sec:
Main Theorems} states that for every $n$ and $0\leq R_i<
\underline C_{n,i}(Q_n)-\frac{\log
|\mathcal S|}{n},\; i=1,2,3 $ (recall, $R_3\triangleq R_1+R_2$) 
and every $\eta>0$  there exists an $N$ and an  $(N,\lceil
2^{NR_1} \rceil,\lceil 2^{NR_1} \rceil)$ code with a
 probability of error $P_{e}(s_0)$
(averaged over the messages) that is less than $\eta$ for all
initial states $s_0$.

%
%
%
%
%

{\it Proof of Theorem \ref{t_inner_bound}:} The proof consists of
the following three steps:
\begin{itemize}
\item Showing that for a fixed $n$ if $ R_i< \underline C_{n,i}(Q_n)-\frac{\log
|\mathcal S|}{n},\; i=1,2,3 $ then there exists $\rho^*$ such
that,
\begin{equation}
F_{n,i}(\rho^*, Q_n)-\rho^* R_i>0, \; i=1,2,3.
\end{equation}
\item We choose $\epsilon< \min_{i\in\{1,2,3\}} F_{n,i}(\rho^*,
Q_n)-\rho^* R_i$ and show that for sufficiently large $N$
\begin{equation}\label{e_pei}
{\bf E}[P_{ei}(s_0)|m_1,m_2]\leq2^{-N([F_{n,i}(\rho^*, Q_n)-\rho^*
R_i]-\epsilon)},\; \forall s_0.
\end{equation}

\item From the last step we deduce the existence of a $(N,\lceil
2^{NR_1} \rceil,\lceil 2^{NR_1} \rceil)$ code s.t.
\begin{equation}
P_{e}(s_0)<\eta,\; \forall s_0.
\end{equation}
\end{itemize}
First step: for any pair $(R_1,R_2)$, we can rewrite eq.
(\ref{e_pfn}) for i=1,2,3 as
\begin{equation}
{\bf E}[P_{ei}(s_0)|m_1,m_2]\leq2^{-N(F_{N,i}(\rho, Q_N)-\rho
R_i-\frac{\log |\mathcal{S}|}{N})}.
\end{equation}
By using (\ref{e_FN_n}), which states that $F_{N,i}(\rho,Q_N)\geq
F_{n,i}(\rho,Q_n)$, we get
\begin{equation}\label{e_pei}
{\bf E}[P_{ei}(s_0)|m_1,m_2]\leq2^{-N(F_{n,i}(\rho, Q_n)-\rho
R_i-\frac{\log |\mathcal{S}|}{N})}.
\end{equation}
Note that $F_{n,i}(\rho,Q_n)$ and therefore
$F_{n,i}(\rho,Q_n)-\rho R$ is continuous in $\rho\in[0,1]$, so
there exists a maximizing $\rho$.
Let us show that if $R_1< \underline C_{n,1}(Q_n)-\frac{\log
|\mathcal S|}{n} $, then $\max_{0\leq \rho \leq
1}[F_{n,1}(\rho,Q_n)-\rho R_1]>0$ 
(the cases $i=2,3$ are identical to $i=1$). Let us define
$\delta\triangleq \underline C_{n,1} - R_1$ . From Lemma
\ref{l_der}, we have
that $E_{n,1}(\rho,Q_N,s_0)$ is zero when $\rho=0$, is a
continuous function of $\rho$, and its derivative at zero with
respect to $\rho$ is equal or greater to $\underline C_{n,1}$,
which satisfies $\underline C_{n,1}\geq R_1+\frac{\log
|\mathcal{S}|}{n}+\frac{\delta}{2}$. Thus, for each state $s_0$
there is a range $\rho>0$ such that
\begin{equation}\label{e_EoN}
E_{n,1}(\rho,Q_N,s_0) - \rho(R_1+\frac{\log |\mathcal{S}|}{n}) > 0.
\end{equation}
Moreover, because the number of states is finite, there exists a
$\rho^*>0$ for which the inequality (\ref{e_EoN}) is true for all
$s_0$. Thus, from the definition of $F_{n,1}(\rho^*,Q_n)$ given in
(\ref{eq_fn}) and from (\ref{e_EoN}),
\begin{equation}
F_{n,1}(\rho^*,Q_n) =-\rho^*\frac{\log
|\mathcal{S}|}{n}+\min_{s_0} E_{n,1}(\rho^*,Q_n,s_0)
>\rho^*R_1, \qquad \forall s_0.
\end{equation}

Second step: We choose a positive number $\epsilon$ such that
$\epsilon< \min_{i\in\{1,2,3\}} F_{n,i}(\rho^*, Q_n)-\rho^* R_i$.
It follows from (\ref{e_pei}) that for every $N$ that satisfies
$N>\frac{\log |\mathcal S|}{\epsilon}$,
\begin{equation}\label{e_pfepsilon}
{\bf E}[P_{ei}(s_0)|m_1,m_2]\leq2^{-N(F_{n,i}(\rho^*, Q_n)-\rho^*
R_i-\epsilon)},
\end{equation}
and according to the first step of the proof the exponent
$F_{n,i}(\rho^*, Q_n,s_0)-\rho^* R_i-\epsilon$ is strictly positive.

Third step: According to the previous step, for all
$\frac{\eta}{3|\mathcal S|+1}>0$ there exists an $N$ such that
${\bf E}[P_{ei}(s_0)|m_1,m_2]\leq \frac{\eta}{3|\mathcal S+1|}$
for all $i\in{1,2,3}$ all $s_0\in\mathcal S$ and all messages.
Since $P_{e}(s_0)=\sum_{i=1}^3P_{ei}(s_0)$, then ${\bf
E}[P_{e}(s_0)|m_1,m_2]\leq \frac{\eta}{|\mathcal S|+1}$;
furthermore ${\bf E}[P_{e}(s_0)]\leq \frac{\eta}{|\mathcal S|+1}$
for all $s_0\in\mathcal S$. By using the Markov inequality, we
have
\begin{eqnarray}
\Pr(P_{e}(s_0)\geq \eta)\leq \frac{1}{|\mathcal S|+1},
\end{eqnarray}
and by using the union bound we have
\begin{eqnarray}
\Pr(P_{e}(s_0)\geq \eta, \text{for some } s_0\in \mathcal S)\leq
\sum_{s_0\in \mathcal S}\Pr(P_{e}(s_0)\geq \eta)= \frac{|\mathcal S
|}{|\mathcal S|+1} < 1.
\end{eqnarray}
Because the probability over the ensemble of codes of having a code
with  error probability (averaged over all messages) that is less
than $\eta$ for all initial states is positive, there must exist at
least one code that has an error probability (averaged over all
messages) that is less than $\eta$ for all initial states. \hfill
\QED

\section{Proof of the Outer Bound (Theorem \ref{t_outer_bound})}
\label{sec: Proof of the outer bound} In this section we prove
Theorem \ref{t_outer_bound}, which states that for any FS-MAC
there exists a distribution
$Q(x_1^n||z_1^{n-1})Q(x_2^n||z_2^{n-1})$ such that the following
inequalities hold:
\begin{eqnarray}\label{e_outer_bound}
R_1 &\leq&  \frac{1}{n}I(X_1^n \to Y^n ||X_2^{n})+\epsilon_n \nonumber \\
R_1 &\leq&  \frac{1}{n}I(X_2^n \to Y^n ||X_1^{n})+\epsilon_n \nonumber \\
R_1+R_2 &\leq& \frac{1}{n}I((X_1,X_2)^n \to Y^{n})+\epsilon_n,
\end{eqnarray}
where $\epsilon_n$ goes to zero as $n$ goes to infinity.

{\it Proof of Theorem \ref{t_outer_bound}:} Let $W_1$ and $W_2$ be
two independent messages, chosen independently and according to a
uniform distribution ${\Pr(W_l=w_l)=2^{-nR_l}},l=1,2$. The input
to the channel from encoder $l$ at time $i$ is $x_{li}$, and is a
function of the message $W_i$ and the arbitrary deterministic
feedback output $z_l^{i-1}(y^{i-1})$.

The following sequence of equalities and inequalities proves that
if a code that achieves rate $R_1$ exists then  the first
inequality holds, i.e., $R_1 \leq \frac{1}{n}I(X_1^n \to Y^n
||X_2^{n})+\epsilon_n$:
\begin{eqnarray}\label{e_con}\label{e_outer_bound_proof}
nR_1 &\stackrel{(a)}{=}& H(W_1) \nonumber \\
&\stackrel{(b)}{=}&H(W_1|W_2) \nonumber \\
&=&I(W_1;Y^n|W_2)+H(W_1|Y^n,W_2) \nonumber \\
& \stackrel{(c)}{\leq} & I(Y^n;W_1|W_2) + 1 + P_e^{(n)}nR \nonumber \\
&=& H(Y^n|W_2)-H(Y^n|W_1,W_2)+1 + P_e^{(n)}nR \nonumber \\
&\stackrel{(d)}{=}& \sum_{i=1}^{n} H(Y_i|Y^{i-1},W_2) - \sum_{i=1}^{n} H(Y_i|W_1,W_2,Y^{i-1})+1 + P_e^{(n)}nR \nonumber \\
&\stackrel{(e)}{=}& \sum_{i=1}^{n} H(Y_i|Y^{i-1},W_2,X_2^i) - \sum_{i=1}^{n} H(Y_i|W_1,W_2,Y^{i-1},X_1^i,X_2^i)+1+ P_e^{(n)}nR \nonumber \\
&\stackrel{(f)}\leq& \sum_{i=1}^{n} H(Y_i|Y^{i-1},X_2^i) - \sum_{i=1}^{n} H(Y_i|Y^{i-1},X_1^i,X_2^i)+1+ P_e^{(n)}nR \nonumber \\
&=& \sum_{i=1}^{n} I(Y_i;X_1^i|Y^{i-1},X_2^i)+1+ P_e^{(n)}nR\nonumber \\
&\leq& I(X_1^n\to Y^n||X_2^n)+1+ P_e^{(n)}nR,
\end{eqnarray}
where,
\begin{itemize}
\item[(a)] and (b) follow from the fact that the messages $W_1$ and
$W_2$ are independent and chosen according to a  uniform
distribution,
\item[(c)] follows from Fano's inequality,
\item[(d)] follows from the chain rule,
\item[(e)] follows from the fact that $x_{1i}$
is a deterministic function given the message $W_1$ and the feedback
$z_1^{i-1}$, where the feedback $z_1^{i-1}$ is a deterministic
function of the output $y^{i-1}$,
\item[(f)] follows from the fact that the random
variables $W_1,W_2,X_1^i,X_2^i,Y^i$ form the Markov chain
$(W_1,W_2)-(X_1^i,X_2^i,Y^{i-1})-Y_i$.
\end{itemize}

Dividing (\ref{e_outer_bound_proof}) by $n$, we conclude that if
there exists a code for which the error probability of decoding
the messages $W_1,W_2$ is $P_e^{(n)}$ then the distribution
$Q(x_1^n||z_1^{n-1})Q(x_2^n||z_2^{n-1})$ induced by the code
satisfies the first inequality of the outer bound theorem where
$\epsilon_n=\frac{1}{n}+ P_e^{(n)}R$. The proofs of the other two
inequalities in (\ref{e_outer_bound}) follow by a completely
analogous sequence of steps as in (\ref{e_outer_bound_proof}): The
proof of the second inequality of the outer bound starts with the
equalities $R_2=H(W_2)=H(W_2|W_1)$ and the third with
$R_1+R_2=H(W_1,W_2)$. \hfill \QED
\begin{corollary}\label{c_outer_bound}
The outer bound given in Theorem \ref{t_outer_bound} implies that
$\liminf {\mathcal R}_n$ is an outer bound for the achievable
region.
\end{corollary}

\begin{proof}
Recall the definition of  ${\mathcal R}_n$ in eq.
(\ref{e_def_Rn}). Let $(R_1,R_2)$ be an achievable rate pair. We
will create a sequence of rate pairs $(R_{1,n},R_{2,n})\in
{\mathcal R}_n$ that converges to $(R_1,R_2)$ and therefore, by
the definition of $\liminf$ of a sequence of sets (given in
Appendix \ref{s_app_supadditive}), $(R_1,R_2)\in \liminf {\mathcal
R}_n$.

If $(R_1,R_2)\in{\mathcal R}_n$ then we choose
$(R_{1,n},R_{2,n})=(R_1,R_2)$. Otherwise we choose the closest
point in ${\mathcal R}_n$ to $R_1,R_2$. Because of inequality
(\ref{e_outer_bound}) the distance
$||(R_{1,n},R_{2,n})-(R_1,R_2)||\leq 2\epsilon_n$ and, therefore,
the sequence $(R_{1,n},R_{2,n})$ converges to $(R_1,R_2)$.
\end{proof}

\section{Capacity Region of the FS-MAC without Feedback}
\label{sec: Capacity region of FS-MAC without feedback} The inner
and outer bounds given in Theorems \ref{t_inner_bound} and
\ref{t_outer_bound} specialize to the case where there is no
feedback, i.e., $z_1,z_2$ are null. Hence, we can use it in order to
extend Gallager's results \cite[Ch. 4]{Gallager68} on the capacity
of indecomposable FSCs to indecomposable FS-MACs. An indecomposable
FS-MAC (FSC) is a FS-MAC (FSC) for which the effect of the initial
state vanishes with time. More precisely:
\begin{definition}
A FS-MAC (FSC) is {\it indecomposable} if, for every $\epsilon>0$,
there exists an $n_0$ such that for $n\geq n_0$,
$|P(s_n|x_1^n,x_2^n,s_0)-P(s_n|x_1^n,x_2^n,s_0')|\leq \epsilon$
for all $s_n$,$x_1^n,x_2^n$ , $s_0$ and $s_0'$.
\end{definition}

Since there is no feedback, according to Lemma
\ref{l_directed_mutaul_if_no_feedback} directed information becomes
mutual information and causal conditioning becomes regular
conditioning in all the expressions in the inner bound (Theorem
\ref{t_inner_bound}) and outer bound (Theorem \ref{t_outer_bound}).

The proof of the capacity region of FS-MAC is based on the
following two lemmas. The first lemma is used for showing that the
difference between the lower bound and the upper bound goes to
zero as $n\to \infty$ and the second lemma, which is proved in
Appendix \ref{s_app_proof_of_lemma_supadditive_Rn}, is used for
showing that the limits exist.
\begin{lemma}\label{l_indecomposible} Let $\{Q(x_1^n)Q(x_2^n)\}_{n\geq1}$
be an arbitrary sequence of input distribution. If the channel is
an indecomposable FS-MAC then the following holds for all $s_0',
s_0''$:
\begin{eqnarray}
\lim_{n\to \infty}
\frac{1}{n}|I(X_1^n;Y^n|X_2^n,s_0')-I(X_1^n;Y^n|X_2^n,s_0'')|&=&
0 \nonumber \\
\lim_{n\to \infty}
\frac{1}{n}|I(X_2^n;Y^n|X_1^n,s_0')-I(X_2^n;Y^n|X_1^n,s_0'')|&=&
0 \nonumber \\
\lim_{n\to \infty}
\frac{1}{n}|I(X_1^n,X_2^n;Y^n|s_0')-I(X_1^n,X_2^n;Y^n|s_0'')|&=&
0.
\end{eqnarray}
\end{lemma}
\begin{proof}
The proof is identical to the proof of Theorem 4.6.4 in
\cite{Gallager68}.
\end{proof}


The following lemma, which is proved in Appendix
\ref{s_app_proof_of_lemma_supadditive_Rn}, establishes the
sup-additivity of $\{\underline {\cal R}_n\}$.
\begin{lemma}\label{l_supadditive_Rn} ({\it sup-additivity of $\underline {\cal R}_n$.
}) For any FS-MAC, the sequence $\{\underline {\cal R}_n\}$, which
is defined in (\ref{e_def_underline_Rn}), is sup-additive, i.e.,
\begin{equation}
(n+l)\underline {\mathcal R}_{n+l}\supseteq n\underline {\mathcal
R}_{n}+l\underline {\mathcal R}_{l},
\end{equation}
and therefore $\lim_{n\to\infty} \underline {\mathcal R}_{n}$
exists. Moreover, for an indecomposable FS-MAC without feedback
$\lim_{n\to\infty} \underline {\mathcal R}_{n}=\lim_{n\to\infty}
{\mathcal R}_{n}$ where ${\mathcal R}_{n}$ is defined
(\ref{e_def_Rn}).
\end{lemma}

 {\it Proof of Theorem
\ref{t_no_feedback}:} Theorem \ref{t_inner_bound} implies that
$\lim_{n\to\infty} \underline {\mathcal R}_{n}$ is achievable, and
Corollary \ref{c_outer_bound} implies that
 $\liminf_{n\to\infty} {\mathcal R}_{n}$ is an
outer bound. Finally, since according to Lemma
\ref{l_supadditive_Rn} the two limits are equal to
$\lim_{n\to\infty} {\mathcal R}_{n}$, the capacity region is given
by the last limit. \hfill \QED

\section{Sufficient Conditions for the Inner and Outer Bounds to Coincide for General Feedback}
\label{sec: Stationary Finite state Markovian MAC with feedback}

\subsection{Stationary Finite state Markovian MAC with feedback} A
stationary finite state Markovian MAC satisfies
\begin{equation}
P(y_i,s_i|x_{1i},x_{2i},s_{i-1})=P(s_i|s_{i-1})P(y_i|s_{i-1},x_{1i},x_{2i}),
\end{equation}
where the initial state distribution is the stationary
distribution $P(s_0)$. In words, the states are not affected by
the channel inputs.

For the stationary Markovian-MAC, the sequence $\{\mathcal R_n\}$ is
sup-additive. It follows from the fact that if we concatenate two
input distributions $Q_{n+k}=Q_nQ_k$, then $I(X_1^{n+k} \to Y^{n+k}
||X_2^{n+k})=I(X_1^{n} \to Y^{n} ||X_2^{n})+I(X_{1,n+1}^{n+k} \to
Y_{n+1}^{n+k} ||X_{2,n+1}^{n+k})$, hence $(n+k)\mathcal
R_{n+k}\supseteq n\mathcal R_{n}+k\mathcal R_{k}  $. According to
Lemma \ref{l_supadditive_region}, the limit exists and is equal to
\begin{equation}
\lim_{n\to\infty}  {\mathcal R}_n= \text{cl} \left(
\bigcup_{n\geq1} {\mathcal R}_n \right).
\end{equation}

Next, we prove Theorem \ref{t_capacity_feedback} that states that
for a Markovian FS-MAC with a stationary ergodic state process,
the inner bound (Theorem \ref{t_inner_bound}) and the outer bound
(Theorem \ref{t_outer_bound}) coincide and therefore the capacity
region is given by $\lim_{n\to \infty} \mathcal R_n$.

{\it Proof of Theorem \ref{t_capacity_feedback}:} Recall that the
inner bound is given in Theorem \ref{t_inner_bound} as $\underline
{\mathcal R}_N$ and the outer bound given in Theorem
\ref{t_outer_bound} and in Corollary \ref{c_outer_bound} as
$\liminf {\mathcal R}_N$. Next we show that the distance between
$\underline {\mathcal R}_N$ and ${\mathcal R}_N$ goes to zero
which implies by Lemma \ref{l_distance_goes_to_zero} that both
limits equal and therefore the capacity region can be written as
$\lim {\mathcal R}_N$.
%

Let us consider a specific input distribution denoted by
$\overline Q(x_1^N||z^{N-1})\overline Q(x_2^N||z^{N-1})$
corresponding to the region of the outer bound  ${\mathcal R}_N$.
Let us now consider an input distribution $\underline Q$ for $n+N$
inputs corresponding to the inner bound $\underline {\mathcal
R}_N$, such that it is arbitrary for the first $n$ inputs and then
it is $\overline Q(x_1^N||z^{N-1})\overline Q(x_2^N||z^{N-1})$.

Now let us show that the term of the inner bound, i.e.
${I_{\underline Q}(X_1^{N} \to Y^{N} ||X_2^{N+n},s_0)}$ and the term
of the outer bound $I_{\overline Q}(X_{1}^{N} \to Y^{N}
||X_{2}^{N})$ are arbitrarily close to each other.
\begin{eqnarray}
\lefteqn{I_{\underline Q}(X_1^{N+n} \to Y^{N+n}
||X_2^{N+n},s_0)}\nonumber\\
&\stackrel{(a)}{\geq}& I_{\underline Q}(X_1^{N+n} \to Y^{N+n} ||X_2^{N+n},S_n,s_0)-\log|\mathcal S| \nonumber \\
&\stackrel{(b)}{\geq}& \sum_{i=n+1}^{N+n} H_{\underline Q}(Y_i|Y^{i-1},X_2^{i},S_n,s_0)-H_{\underline Q}(Y_i|Y^{i-1},X_2^{i},X_1^{i},S_n,s_0)-\log|\mathcal S| \nonumber \\
&\stackrel{(c)}{\geq}& \sum_{i=n+1}^{N+n} H_{\underline Q}(Y_i|Y_{n+1}^{i-1},X_{2,n+1}^{i},S_n,s_0)-H_{\underline Q}(Y_i|Y_{n+1}^{i-1},X_{2,n+1}^{i},X_{1,n+1}^{i},S_n,s_0)-\log|\mathcal S| \nonumber \\
&=& I_{\underline Q}(X_{1,n+1}^{N+n} \to Y_{n+1}^{N+n} ||X_{2,n+1}^{N+n},S_n,s_0)-H(S_n) \nonumber \\
&\stackrel{(d)}{\geq}& I_{\underline Q}(X_{1,n+1}^{N+n} \to Y_{n+1}^{N+n} ||X_{2,n+1}^{N+n},S_n)(1-\delta)-\log|\mathcal S| \nonumber \\
&\stackrel{}{\geq}& I_{\underline Q}(X_{1,n+1}^{N+n} \to Y_{n+1}^{N+n} ||X_{2,n+1}^{N+n},S_n)-\delta(N+n)\log|\mathcal Y|-\log|\mathcal S| \nonumber \\
&\stackrel{(e)}{\geq}& I_{\underline Q}(X_{1,n+1}^{N+n} \to Y_{n+1}^{N+n} ||X_{2,n+1}^{N+n})-\delta(N+n)\log|\mathcal Y|-2\log|\mathcal S| \nonumber \\
&\stackrel{(f)}{\geq}& I_{\overline Q}(X_{1}^{N} \to Y^{N}
||X_{2}^{N})-\delta(N+n)\log|\mathcal Y|-2\log|\mathcal S|,
\end{eqnarray}
where
\begin{itemize}
\item[(a)] follows from Lemma \ref{l_diff_cond_S} that states that conditioning on $S_n$ can differ at most by $\log|\mathcal S|$,
\item[(b)] follows from omitting the first $n$ elements in the sum that
defines directed information,
\item[(c)]follows from the fact that conditioning decreases entropy,
\item[(d)]follows from the fact that the Markov chain is ergodic,
hence for any $\delta>0$, there exists an $n$ such that
$|P(s_n|s_0)-P(s_n)|\leq \delta$ for any $s_0\in \mathcal S$ and
$s_n\in \mathcal S$, where $P(s_n)$ is the stationary distribution
of $s_n$,
\item[(e)]  follows from Lemma \ref{l_diff_cond_S} that states that conditioning on $S_n$ can differ by at most $\log |\mathcal
S|$,
\item[(f)] follows from the stationarity of the channel.
\end{itemize}

Dividing both sides by $N+n$ we get that for any $s_0$,
\begin{eqnarray}\label{e_diff_upper_lower}
\frac{1}{N+n}I_{\underline Q}(X_1^{N+n} \to Y^{N+n}
||X_2^{N+n},s_0)-\frac{1}{N+n}I_{\overline Q}(X_{1}^{N} \to Y^{N}
||X_{2}^{N})\geq -\delta(1+\frac{n}{N})\log|\mathcal
Y|-2\frac{\log|\mathcal S|}{N+n}
\end{eqnarray}

Inequality (\ref{e_diff_upper_lower}) shows that the difference
between the upper bound region and the lower bound is arbitrarily
small for $N$ large enough and, hence, in the limit the regions
coincide.\hfill\QED


\subsection{Finite State Markovian MAC with limited ISI}
In this subsection we consider a MAC inspired by Kim's
 point-to-point channel \cite{Kim07_feedback}. The conditional probability of
the MAC is given by
\begin{equation}\label{e_channel_limited_isi}
P(y_i,z_i|x_{1}^{i},x_{2}^{i},z_{i-1})=P(z_i|z_{i-1})P(y_i|z_{i-1},x_{1,i-m}^i,x_{2,i-m}^i),
\;i=1,2,3,...
\end{equation}
where the distribution of $Z_0$ is the stationary distribution
$P(z_0)$, and there is also some initial distribution
$P(x_{-m+1},...,x_0)$.

This channel is a FS-MAC where the state at time $i$ is
$(z_{i-1},x_{1,i-m}^{i-1},x_{2,i-m}^{i-1})$ and therefore the
inner bound (Theorem \ref{t_inner_bound}) and the outer bound
(Theorem \ref{t_outer_bound}) apply to this channel. Theorem
\ref{t_capacity_feedback} also holds for this kind of channels,
namely, the capacity region is given by $\lim_{n\to\infty}\mathcal
R_n.$ The proof is very similar, the only difference being that
the input $\underline Q$ for $n+N$ inputs is constructed slightly
differently: it is arbitrary for the first $n-m$ inputs, then it
is as the initial distribution $P(x_{-m+1},...,x_0)$, and then it
is $\overline Q(x_1^N||z^{N-1})\overline Q(x_2^N||z^{N-1})$.


It is also possible to represent the channel with an alternative
law,  identical to the law of the channel given in eq.
(\ref{e_channel_limited_isi}) for $i\geq m+1$ but for  $i\leq m$
the output $y_i$ is not influenced by the input and is, with
probability 1, a particular output $\phi\in \mathcal Y$. Let us
define $\mathcal R_n^{\phi}$ similarly as $\mathcal R_n$ but with
the alternative law for the channel. On one hand, it is clear that
$ \mathcal R_n^{\phi}\subseteq \mathcal R_n$ for all $n$, and on
the other hand the difference between $ \mathcal R_n^{\phi}$ and $
\mathcal R_n$ is at most $m\log \mathcal Y$ because it is possible
to use the distribution of the first $m$ inputs, $Q(x_1^m)$, to
create a desired initial distribution and then use the same input
as in $\mathcal R_n$. Hence,
\begin{equation}
\lim_{n\to \infty} \mathcal R_n^{\phi}= \lim_{n\to \infty} \mathcal
R_n.
\end{equation}
The advantage of analyzing $\mathcal R_n^{\phi}$ rather than
analyzing $\mathcal R_n$ is that the sequence $nR_n^{\phi}$ is
sup-additive, i.e. $(n+l)\mathcal R_{n+l}^{\phi}\supseteq
n\mathcal R_{n}^{\phi}+l\mathcal R_l^{\phi}$, and according to
Lemma \ref{l_supadditive_region}, $\lim_{n\to \infty} \mathcal
R_{n}^{\phi}=\text{cl} \left
(\bigcup_{n\geq1}R_{n}^{\phi}\right)$. Hence, we can conclude that
Theorem \ref{t_capacity_zero} holds for this channel too, namely,
if the capacity of the Finite state Markovian MAC with limited ISI
is zero without feedback then it is zero also in the presence of
feedback.


\section{Conclusions and Future Directions} \label{sec: conclusions}
In this paper we have shown that directed information and causal
conditioning emerge naturally in characterizing the capacity region
of FS-MACs in the presence of a time-invariant feedback. The
capacity region is given as a `multi-letter' expression and it is a
first step toward deriving useful concepts in communication. For
instance, we use this characterization in order to show that for a
stationary and ergodic Markovian channel, the capacity is zero if
and only if the capacity with feedback is zero. Further, we identify
 FS-MACs for which feedback does not enlarge the capacity region
and for which source-channel separation holds.

For the point-to-point channel with feedback, recent work has shown
that, for some families of channels such as unifilar channels
\cite{Permuter06_trapdoor_submit} or the additive Gaussian where the
noise is ARMA \cite{Kim07_feedback}, the directed information
formula can be computed and, further, can lead to the development of
capacity achieving coding schemes. One future direction is to use
the characterizations developed in this paper to explicitly compute
the capacity regions of classes of MACs with memory and feedback
(other than the multiplexer followed by a point-to-point channel),
and to find optimal coding schemes.


\newpage
 \appendices
\section{Proof of Lemma \ref{l_mutual_becomes_directed_causal}}\label{s_app_lemma_proof_mutual_becomes_directed}
Recall that Lemma \ref{l_mutual_becomes_directed_causal} states
that if
\begin{equation}\label{e_assumption_conditional_indep}
Q(x_1^N,x_2^N||y^{N-1})=Q(x_1^N||y^{N-1})Q(x_2^N||y^{N-1}),
\end{equation}
 then
\begin{equation} \mathcal
I(Q(x_1^N,x_2^N||y^{N-1});P(y^N||x_1^N,x_2^N))=I(X_1^N\to
Y^N||X_2^N).
\end{equation}

\begin{proof}
The following sequence of equalities proves the lemma.
\begin{eqnarray}
\lefteqn{ \mathcal
I(Q(x_1^N,x_2^N||y^{N-1});P(y^N||x_1^N,x_2^N))}\nonumber
\\
&\stackrel{(a)}{=}& \mathcal I(Q(x_1^N||y^{N-1})Q(x_2^N||y^{N-1});P(y^N||x_1^N,x_2^N))\nonumber \\
&\stackrel{(b)}{=}&\sum_{y^N,x_1^N,x_2^N}Q(x_1^N||y^{N-1})Q(x_2^N||y^{N-1})P(y^N||x_1^N,x_2^N)\frac{P(y^N||x_1^N,x_2^N)}{\sum_{{x'}_1^N}Q({x'}_1^N||y^{N-1})P(y^N||{x'}_1^N,x_2^N)}\nonumber\\
&\stackrel{(c)}{=}&\sum_{y^N,x_1^N,x_2^N}P(x_1^N,x_2^N,y^N)\frac{P(y^N||x_1^N,x_2^N)}{\sum_{{x'}_1^N}Q({x'}_1^N||y^{N-1})P(y^N||{x'}_1^N,x_2^N)}\nonumber\\
&\stackrel{}{=}&{\bf E} \left[\frac{P(y^N||x_1^N,x_2^N)}{\sum_{{x'}_1^N}Q({x'}_1^N||y^{N-1},x_2^{N})P(y^N||{x'}_1^N,x_2^N)}\right]\nonumber\\
&\stackrel{}{=}&{\bf E} \left[\frac{Q({x}_2^N||y^{N-1})P(y^N||x_1^N,x_2^N)}{Q({x}_2^N||y^{N-1})\sum_{{x'}_1^N}Q({x'}_1^N||y^{N-1},x_2^{N})P(y^N||{x'}_1^N,x_2^N)}\right]\nonumber\\
&\stackrel{}{=}&{\bf E} \left[\frac{Q({x}_2^N||y^{N-1})P(y^N||x_1^N,x_2^N)}{\sum_{{x'}_1^N}P(y^N,{x'}_1^N,x_2^N)}\right]\nonumber\\
&\stackrel{}{=}&{\bf E} \left[\frac{Q({x}_2^N||y^{N-1})P(y^N||x_1^N,x_2^N)}{P(x_2^N,y^N)}\right]\nonumber\\
&\stackrel{}{=}&{\bf E} \left[\frac{P(y^N||x_1^N,x_2^N)}{P(y^N||x_2^N)}\right]\nonumber\\
&\stackrel{(d)}{=}&I(X_1^N\to Y^N||X_2^N)\nonumber\\
\end{eqnarray}
\begin{itemize}
\item[(a)] follows from the assumption given in eq.
(\ref{e_assumption_conditional_indep}).
\item[(b)] follows from the definition of the functional $\mathcal
I(Q;P)$ given in eq. (\ref{e_calI_def}).
\item[(c)] follows from Lemma
\ref{l_joint_causal_condition_decomposition} that states that
$P(x_1^N,x_2^N,y^N)=Q(x_1^N,x_2^N||y^{N-1})P(y^N||x_1^N,x_2^N)$
and the assumption given in
(\ref{e_assumption_conditional_indep}).
\item[(d)] follows from the definition of directed information.
\end{itemize}
\end{proof}

\section{Proof of Lemma \ref{l_directed_mutaul_if_no_feedback}}\label{s_app_lemma_proof_directed_mutual_no_feedback}
Lemma \ref{l_directed_mutaul_if_no_feedback} states that if
\begin{equation}\label{e_assumption_indep}
Q(x_1^N,x_2^N||y^{N-1})=Q(x_1^N)Q(x_2^N),
\end{equation}
 then
\begin{eqnarray}
I(X_1^N;Y^N|X_2^N) = I(X_1^N \to Y^N||X_2^N).
\end{eqnarray}

\begin{proof}
The following sequence of equalities proves the lemma.
\begin{eqnarray}
I(X_1^N;Y^N|X_2^N) &\stackrel{}{=}&  \mathbf E\left[ \log \frac{
P(Y^N,X_1^N|X_2^N)}{ P(Y^N|X_2^N)Q(X_1^N|X_2^N)}\right]\nonumber\\
&\stackrel{(a)}{=}&  \mathbf E\left[ \log \frac{
P(Y^N,X_1^N,X_2^N)}{
P(Y^N,X_2^N)Q(X_1^N|X_2^N)}\right]\nonumber\\
&\stackrel{(b)}{=}&  \mathbf E\left[ \log \frac{
Q(X_1^N,X_2^N||Y^{N-1})P(Y^N||X_1^N,X_2^N)}{
P(Y^N||X_2^N)Q(X_2^N||Y^{N-1})Q(X_1^N|X_2^N)}\right]\nonumber\\
&\stackrel{(c)}{=}&  \mathbf E\left[ \log \frac{
Q(X_1^N)Q(X_2^N)P(Y^N||X_1^N,X_2^N)}{
P(Y^N||X_2^N)Q(X_2^N)Q(X_1^N)}\right]\nonumber\\
&\stackrel{}{=}&  \mathbf E\left[ \log \frac{ P(Y^N||X_1^N,X_2^N)}{
P(Y^N||X_2^N)}\right]\nonumber\\
&\stackrel{}{=}&  I(X_1^N \to Y^N||X_2^N).
\end{eqnarray}
\begin{itemize}
\item[(a)] follows from multiplying the numerator and denominator
by $P(x_2^N)$.
\item[(b)] follows from decomposing the joint distributions $P(y^N,x_1^N,x_2^N)$ and $P(Y^N,X_2^N)$ into causal conditioning distribution by using Lemma
\ref{l_joint_causal_condition_decomposition}.
\item[(c)] follows from the fact that the assumption of the lemma given in
(\ref{e_assumption_indep}) implies that
$Q(X_1^N,X_2^N)=Q(X_1^N)Q(X_1^N)$. This can be obtained by
multiplying both sides of (\ref{e_assumption_indep}) by
$P(y^n||x_1^n,x_2^n)$ and then summing over all $y^n\in \mathcal
Y^n$.
\end{itemize}
\end{proof}
\section{Proof of Lemma
\ref{l_zero_iff}}\label{s_app_proof_of_lemma_zero_iff}

Lemma \ref{l_zero_iff} states that
\begin{equation}\label{eqn:0iff}
\max_{Q(x_1^n||y^{n-1})Q(x_2^n||y^{n-1})} I(X_1^n,X_2^n\to Y^n)=0
 \iff
\max_{Q(x_1^n)Q(x_2^n)} I(X_1^n,X_2^n\to Y^n)=0,
\end{equation}
and each condition also implies that $P(y^n||x_1^n,x_2^n)=P(y^n)$
for all $x_1^n,x_2^n$.

\begin{proof}
Proving the direction $\Longrightarrow$ is trivial since
\begin{equation}
\max_{Q(x_1^n||y^{n-1})Q(x_2^n||y^{n-1})} I(X_1^n,X_2^N\to Y^n)
\geq \max_{Q(x_1^n)Q(x_2^n)} I(X_1^n,X_2^n\to Y^n). \end{equation}
For the other direction, $\Longleftarrow$, we have the assumption
that $I(X_1^n,X_2^n\to Y^n)=0$ for all input distributions
$Q(x_1^n)Q(x_2^n)$, and in particular for the case that $X_1^n$
and $X_2^n$ are uniformly distributed over their alphabets.
Directed information can be written as a Kullback Leibler
divergence, i.e.,
\begin{equation}\label{eqn:directed_as_divergence}
\sum_{x_1^n,x_2^n,y^n}Q(x_1^n)Q(x_1^n)P(y^n||x_1^n,x_2^n)\log
\frac{Q(x_1^n)Q(x_1^n)P(y^n||x_1^n,x_2^n)}{P(y^n)Q(x_1^n)Q(x_2^n)}=0
\end{equation}
and by using the fact that if the Kullback Leibler divergence
$D(P||Q)\triangleq \sum_{x\in \mathcal X} P(x)\log
\frac{P(x)}{Q(x)}$ is zero, then $P(x)=Q(x)$ for all $x\in
\mathcal X$, we conclude that (\ref{eqn:directed_as_divergence})
implies that $P(y^n||x_1^n,x_2^n)=P(y^n)$ for all $x_1^n\in
\mathcal X_1^n$ and all $x_2^n\in \mathcal X_2^n$. It follows that
\begin{eqnarray}
\max_{Q(x_1^n||y^{n-1})Q(x_2^n||y^{n-1})}I(X_1^n,X_2^n{\to}Y^n) & =
& \max_{Q(x_1^n||y^{n-1})Q(x_2^n||y^{n-1})} {\bf E}\left[ \log
\frac{P(Y^n||X_1^n,X_2^n)}{P(Y^n)}\right] \nonumber \\
 & = & \max_{Q(x_1^n||y^{n-1})Q(x_2^n||y^{n-1})} {\bf E}[0]=0.
\end{eqnarray}
\end{proof}

\section{Sup-additivity and Convergence of 2D
regions}\label{s_app_supadditive}

Let $A,B$ be sets in $\mathbb{R}^2$, i.e., $A$ and $B$ are sets of
2D vectors. The sum of two regions is denoted as $A+B$ and defined
as
\begin{equation}
A+B=\{{ {\bf a}+\bf b:\;{\bf a}}\in A, { \bf b}\in B\},
\end{equation}
and multiplication of a set $A$ with a scalar $c$ is defined as
\begin{equation}
cA=\{c {\bf a}: \;{\bf a}\in A\}.
\end{equation}



A sequence $\{A_n\},\; {n=1,2,3,...},$ of 2D regions is said to
{\it converge} to a region $A$, written $A=\lim A_n$ if
\begin{equation}
\lim\sup A_n= \lim\inf A_n=A
\end{equation}
where
\begin{eqnarray}\label{e_def_sup_inf_sets}
\lim\inf A_n&=&\left\{ {\bf a}:{\bf a}=\lim {\bf a}_n, {\bf a}_n\in A_n \right\},\nonumber \\
\lim\sup A_n&=&\left\{ {\bf a}:{\bf a}=\lim {\bf a}_k, {\bf
a}_k\in A_{n_k} \right\},
\end{eqnarray}
and $n_k$ denotes an arbitrary increasing subsequence of the
integers. An alternative and equivalent definition of $\lim\sup$
and $\lim\inf$ is given by  $\lim\sup A_n=\bigcap_{n\geq 1} {\text
{cl}}\left(\bigcup_{m\geq n}A_m\right)$ and $\lim\inf
A_n=\bigcup_{n\geq 1} {\text {cl}}\left(\bigcap_{m\geq n}A_m
\right)$. For more details on convergence of sets in finite
dimensions  see \cite{Salinetti:1979:CSC}.



Let $\overline A$ denote
\begin{equation}
\overline A=\text {cl}\left( \bigcup_{n\geq1} A_n \right).
\end{equation}


We say that a sequence $\{A_n\}_{n\geq1}$ is {\it bounded} if
$\sup\{||{\bf a}||: {\bf a}\in \overline A\}<\infty$ where
$||\cdot||$ denotes a norm in $\mathbb{R}^2$.

\begin{lemma}\label{l_supadditive_region}
Let $A_n$, $n=1,2,...$, be a bounded sequence of sets  in
$\mathbb{R}^2$ that includes the origin, i.e. $(0,0)$. If $nA_n$
is sup-additive, i.e., for all $n\geq1$ and all $N>n$
\begin{eqnarray}\label{e_supadditive_property}
NA_N\supseteq nA_n+(N-n)A_{N-n}
\end{eqnarray}
then
\begin{eqnarray}
\lim_{n\to \infty} A_n=\overline A.
\end{eqnarray}
\end{lemma}

\begin{proof}
From the definitions we have $\overline A \supseteq \lim\sup
A_n\supseteq \lim\inf A_n$. Hence it is enough to show that
$\overline A \subseteq \lim\inf A_n$.

Let $ {\bf a}$ be a point in $\overline A$. Then for every
$\epsilon>0$ there exists an $n$ and a point ${ {\bf a}}_\epsilon$
such that ${ {\bf a}}_\epsilon\in A_n$ and $||{ {\bf a}-{\bf
a}}_\epsilon||\leq \epsilon$. By induction we prove that for any
integer $m\geq2$, ${A_n}\subseteq A_{mn}$, and this implies that
${ {\bf a}_\epsilon}\in A_{mn}$.
 For $m=2$ we choose
$N=2n$ and we get that
\begin{eqnarray}\label{e_ind1}
A_{2n}\supseteq \frac{A_n}{2}+\frac{A_n}{2}\supseteq A_n.
\end{eqnarray}
Now assume that it holds for $m-1$ and let us show that it holds
for $m$.
\begin{eqnarray}\label{e_ind2}
A_{mn}\supseteq \frac{A_n}{m}+\frac{(m-1)A_{(m-1)n}}{m}\supseteq
\frac{A_n}{m}+\frac{(m-1)A_{n}}{m}\supseteq {A_n}.
\end{eqnarray}
Now, for any $N>n$, we can represent $N$ as $mn+j$ where $0\leq
j\leq n-1$, hence
\begin{eqnarray}\label{e_Amn+j}
A_{mn+j}\supseteq \frac{j}{mn+j}A_j+\frac{mn}{mn+j}A_{mn}.
\end{eqnarray}
Because ${ {\bf a}}_\epsilon$ is in $A_n$, then it implies that it
is in $A_{mn}$ too. Following (\ref{e_Amn+j}) and the fact that
$(0,0)\in A_j$ we obtain
\begin{equation}\label{e_amn+j} \frac{mn}{mn+j}{ {\bf
a}}_\epsilon
\in A_{mn+j}.\end{equation} For any $\delta>0$ and for any
$N\geq\frac{n}{\delta}$ we conclude the existence of an element
in $A_N$ for which the distance from
 ${\bf a}$ can be upper-bounded by
\begin{equation}
\left \|\frac{mn}{mn+j}{\bf a}_\epsilon-{ {\bf a}}\right
\|=\left\|{ {\bf a}}_\epsilon-{ {\bf a}}-\frac{j}{mn+j}{ {\bf
a}}_\epsilon\right\|\leq ||{ {\bf a}}_\epsilon-{ {\bf
a}}||+\delta||{ {\bf a}}_\epsilon||\leq \epsilon+ \delta||{ {\bf
a}}_\epsilon||.
\end{equation}
%
Because $\epsilon$ and $\delta$ are arbitrarily small we can find
a sequence of points ${ {\bf a}}_n\in A_n$ that converges to $
{\bf a}$ and therefore ${\bf a}\in\lim\inf A_n$, which implies
that $\overline A \subseteq \lim\inf A_n$.
\end{proof}
\begin{corollary}\label{c_convexity_limit_sup}
For a sup-additive sequence, as defined  in Lemma
\ref{l_supadditive_region}, the limit is convex.
\end{corollary}
This corollary follows immediately from  the definition of the
sup-additivity property, eq. (\ref{e_supadditive_property}) where
$n=\alpha N$, where $0<\alpha<1$, and $N$ goes to infinity.

 The  {\it (Hausdroff) distance}
between two sets $A$ and $B$, is defined as
\begin{equation}
d(A,B)=\max \{\sup[d({\bf a},B:\;{\bf a}\in A], \sup [d({\bf
b},A):\;{\bf b}\in B]\},
\end{equation}
 where the distance between a set
$A$ and a point $ {\bf b}$ is given by,
\begin{equation}
d({\bf  b},A)=\inf_{\bf a} [||{\bf a}-{\bf b}||:{\bf a} \in A]
\end{equation}

\begin{lemma}\label{l_distance_goes_to_zero}
If $\lim_{n\to \infty} d(A_n,B_n)=0$ then
\begin{eqnarray}\label{e_limsupinf_equality}
\limsup A_n &=&\limsup B_n, \nonumber \\
\liminf A_n &=&\liminf B_n.
\end{eqnarray}
\end{lemma}
\begin{proof}
The proof is straightforward. Given a sequence $\{{\bf a}_{k}\}\in
A_{n_k}$ that converges to ${\bf a}$, we construct a sequence
$\{\bf b_k\}$ by finding a point in $B_{n_k}$ that is at a
distance less than
$\frac{1}{k}+ d({\bf a}_k,B_{n_k})$. 
Since the distance between the sets goes to zero, $\lim {\bf
b}_k=\lim {\bf a}_k=\bf a$ and from the definitions of limits of
sets, it implies that (\ref{e_limsupinf_equality}) holds.
\end{proof}

\section{Proof of Lemma
\ref{l_supadditive_Rn}}\label{s_app_proof_of_lemma_supadditive_Rn}

Recall the definition of $\underline {\mathcal R}_n$ and
${\mathcal R}_n$  in (\ref{e_def_underline_Rn}) and
(\ref{e_def_Rn}) respectively.

Lemma \ref{l_supadditive_Rn} states that
\begin{equation}\label{e_app_supadditiveRn}
(n+l)\underline {\mathcal R}_{n+l}\supseteq n\underline {\mathcal
R}_{n}+l\underline {\mathcal R}_{l}.
\end{equation}
and for an indecomposable FS-MAC without feedback
$\lim_{n\to\infty} \underline {\mathcal R}_n=\lim_{n\to\infty}
{\mathcal R}_n$.

 {\it Proof of Lemma \ref{l_supadditive_Rn}:} We notice that
if a sequence of sets is sup-additive then the sequence of the
convex hull of the sets is also sup-additive. Hence, it is enough
to prove the sup-additivity of the sequence $\underline {\mathcal
R}_n$ without the appearance of the random variable $W$ that its
role is to convexify the regions.

The set ${\underline {\mathcal
R}_{n}}$ is defined 
 by three
expressions that involve directed information. Because each
expression is sup-additive the whole set is sup-additive. We prove
that the first expression, i.e. $\min_{s_0} I(X_1^n \to Y^n
||X_2^{n},s_0)-\log |\mathcal S|$ is sup-additive (the proofs of
the supper-additivity of the other expressions are similar and
therefore omitted).

\begin{eqnarray}
\min_{s_0} \lefteqn{ I(X_1^{n+l} \rightarrow Y^{n+l}||X_2^{n+l},s_0)} \nonumber \\
& \stackrel{(a)}{\geq} &  \min_{s_0} \sum_{i=1}^{n} I(Y_i;X_1^i|Y^{i-1},X_2^i,s_0)+ \min_{s_0} \sum_{j=n+1}^{n+l}I(Y_j;X_1^j|Y^{j-1},X_2^{j},s_0)    \nonumber \\
& \stackrel{(b)}{\geq} & I(X_1^n \to Y^n
||X_2^{n},s_0) +   \sum_{j=n+1}^{n+l}I(Y_j;X_{1,n+1}^j|Y^{j-1},X_2^{j},s_0)  \nonumber \\
& \stackrel{(c)}{\geq} & I(X_1^n \to Y^n
||X_2^{n},s_0) +  \sum_{j=n+1}^{n+l}I(Y_j;X_{1,n+1}^j|Y^{j-1},X_2^{j},S_n,s_0)-\log |\mathcal{S}|  \nonumber \\
& \stackrel{}{=} &\min_{s_0} I(X_1^n \to Y^n
||X_2^{n},s_0)+  \min_{s_0} \sum_{s_n} P(s_n|s_0) \sum_{j=n+1}^{n+l}I(Y_j;X_{1,n+1}^j|Y^{j-1},X_{2,n+1}^{j},s_n)-\log |\mathcal{S}|  \nonumber \\
& \stackrel{}{\geq} &\min_{s_0} I(X_1^n \to Y^n
||X_2^{n},s_0) +  \min_{s_n} \sum_{j=n+1}^{n+l}I(Y_j;X_{1,n+1}^j|Y_{n+1}^{j-1},X_{2,n+1}^{j},s_n)-\log |\mathcal{S}|  \nonumber \\
& \stackrel{(d)}{=} &\min_{s_0} I(X_1^n \to Y^n ||X_2^{n},s_0)+
\min_{s_0} I(X_1^l \to Y^l ||X_2^{l},s_0)-\log |\mathcal{S}|.
\end{eqnarray}

\begin{itemize}
\item[(a)] follows the definition of the directed information the fact that $\min_s[f(s)+g(s)]\geq \min_s f(s)+\min_s
g(s)$,
\item[(b)] follows the fact that $I(X;Y,Z)\geq
I(X;Y)$,
\item[(c)] follows Lemma \ref{l_diff_cond_S} that states that conditioning by $S_n$ can differ by at most $\log |\mathcal S|$,
\item [(d)] follows from the stationarity of the channel.
\end{itemize}

According to Lemma \ref{l_supadditive_region}, since the sequence
$\{\underline {\mathcal R}_n\}$ is sup-additive the limit exists. In
the rest of the proof we show that $\lim_{n\to\infty} \underline
{\mathcal R}_n=\lim_{n\to\infty} {\mathcal R}_n$. The terms of the
region $\underline {\mathcal R}_n$ have an auxiliary random variable
$W$ whose only role is to convexify the region. Let us denote
$\underline {\mathcal R}_n^o$ the same region as $\underline
{\mathcal R}_n$ where $W$ is restricted to be null. We show first
that restricting $W$ to being null does not influence the limit,
i.e., $ \lim_{n\to \infty} \underline{\mathcal R}_n=\lim_{n\to
\infty} \underline{\mathcal R}_n^o$. In the first half of the proof
we showed that $\underline {\mathcal R}_n^o$ is sub-additive. Using
this fact, we show now, that any convex combination with rational
weights $(\frac{l}{k},\frac{k-l}{k})$ of any two points from
$\underline {\mathcal R}_n^o$ is in $\underline {\mathcal
R}_{kn}^o$.
\begin{eqnarray}\label{e_inclusions}
\underline {\mathcal R}_{kn}^o\supseteq \frac{l}{k}\underline
{\mathcal R}_{ln}^o+\frac{k-l}{k}\underline {\mathcal
R}_{(k-l)n}^o\supseteq \frac{l}{k}\underline {\mathcal
R}_{n}^o+\frac{k-l}{k}\underline {\mathcal R}_{n}^o
\end{eqnarray}
The left and the right inclusions in (\ref{e_inclusions}) are due to
the sup-additivity of $\underline {\mathcal R}_{n}^o$. The left
inclusion is from the definition of the sup-additivity and the right
is due to the fact that sup-additivity of $\underline {\mathcal
R}_{n}^o$ also implies that for any two positive integers $m,n$,
$\underline {\mathcal R}_{mn}^o \supseteq \underline {\mathcal
R}_{n}^o$ (This is shown by induction in
(\ref{e_ind1},\ref{e_ind2})). From (\ref{e_inclusions}) we can
deduce that for any $\epsilon>0$ we can find a $k(\epsilon)$ such
that $\underline {\mathcal R}_{n} \subseteq \underline {\mathcal
R}_{nk}^o+\epsilon$. This fact, together with the trivial fact that
$ \underline {\mathcal R}_{n}\supseteq \underline {\mathcal
R}_{n}^o$, and the fact that the limits of both sequences exist,
allow us to deduce that the limits are the same, i.e., $ \lim_{n\to
\infty} \underline{\mathcal R}_n=\lim_{n\to \infty}
\underline{\mathcal R}_n^o$.
%
%
%


We conclude the proof by showing that, for any input distribution
$Q(x_1^n)Q(x_2^n)$, the difference between the terms in the
inequalities of $\{\underline {\mathcal R}_n^o\}$ and $\{
{\mathcal R}_n\}$ goes to zero as $n\to \infty$, hence the
distance between the sets of the sequences goes to zero as $n\to
\infty$ and, by Lemma \ref{l_distance_goes_to_zero}, the limits of
the sequences are the same.
\begin{eqnarray}
\lefteqn{\lim_{n\to\infty} \frac{1}{n}\left |I(X_1^n \to Y^n
||X_2^{n})-\min_{s_0} I(X_1^n \to Y^n
||X_2^{n},s_0)+\log |\mathcal S|\right|}\nonumber \\
&\stackrel{(a)}{\leq}& \lim_{n\to\infty} \frac{1}{n}\left|I(X_1^n
\to Y^n ||X_2^{n},S_0)-\min_{s_0} I(X_1^n \to Y^n
||X_2^{n},s_0)+\log
|\mathcal S|\right|+\log |\mathcal S| \nonumber \\
&\stackrel{}{=}& \lim_{n\to\infty} \frac{1}{n}\left[I(X_1^n \to
Y^n ||X_2^{n},S_0)- \min_{s_0}
I(X_1^n \to Y^n ||X_2^{n},s_0))\right]\nonumber \\
&\stackrel{(b)}{\leq}& \lim_{n\to\infty}
\frac{1}{n}\left[\max_{s_0}I(X_1^n \to Y^n ||X_2^{n},s_0)-
\min_{s_0}
I(X_1^n \to Y^n ||X_2^{n},s_0))\right]\nonumber \\
&\stackrel{(c)}{=}&0
\end{eqnarray}
\begin{itemize}
\item[(a)] follows from Lemma \ref{l_diff_cond_S} and the triangle
inequality.
\item [(b)] follows from the fact that
$\max_{s_0}I(X_1^n \to Y^n ||X_2^{n},s_0)\geq I(X_1^n \to Y^n
||X_2^{n},S_0)$. 
\item[(c)] follows from Lemma \ref{l_indecomposible} that states this equality for indecomposable FS-MAC without
feedback (recall also that directed information equals mutual
information in the absence of feedback).
\end{itemize}
\hfill \QED

\section{Proof of Theorem \ref{t_MLB}} \label{s_app_proof_MLB}
\begin{eqnarray}\label{e_pem}
\mathbf E[P_{e1}]& = & \sum_{y^N}\sum_{x_1^N,x_2^N}
P(x_1^N,x_2^N,y^N)P[error1|m_1,m_2,x_1^N,x_2^N,y^N]
\nonumber \\
& = & \sum_{y^N}\sum_{x_1^N,x_2^N}
Q(x_1^N||z_1^{N-1})Q(x_2^N||z^{N-1})P(y^N||x_1^N,x_2^N)P[error1|m_1,m_2,x^N,y^N],
\end{eqnarray}
where $P[error1|m_1,m_2,x^N,y^N]$ is the error probability of
decoding $m_1$ given that $m_2$ is decoded correctly. Throughout
the remainder of the proof we fix the message $m_1,m_2$. For a
given tuple $(m_1,m_2,x_1^N,x_2^N,y^N)$ define the event
$A_{m_1'}$, for each $m_1'\neq m_1$, as the event that the message
$m_1'$ is selected in such a way that
$P(y^N|m_1',m_2)>P(y^N|m,m_2)$ which is the same as
$P(y^N||{x'}_1^N,x_2^N)> P(y^N||x_1^N,x_2^N)$ where ${x'}_1^N$ is
a shorthand notation for $x_1^N(m_1',z^{N-1}(y^{N-1}))$ and
$x_i^N$ is a shorthand notation for
$x_l^N(m_l,z_l^{N-1}(y^{N-1}))$ for $l=1,2$. From the definition
of $A_{m_1'}$ we have
\begin{eqnarray}\label{e_pAm}
P(A_{m_1'}|m_1,m_2,x_1^N,x_2^N,y^N)&=&\sum_{x'^N}
Q({x'}_1^{N}||z^{N-1}) \cdot {\bf 1}
[P(y^N||{x'}_1^N,x_2^N)> P(y^N||x_1^N,x_2^N )]\nonumber \\
& \leq & \sum_{x'^N}  Q ({x'}_1^N||z^{N-1}) \left[\frac{
P(y^N||{x'}_1^N,x_2^N)}{P(y^N||x_1^N,x_2^N)}\right]^s; \qquad
\text{any } s>0
\end{eqnarray}
where ${\bf 1}(x)$ denotes the indicator function.
\begin{eqnarray}\label{e_MLB2}
P[error1|m_1,m_2,x_1^N,x_2^N,y^N] & = & P(\bigcup_{m'\neq m}
A_{m_1'}|m_1,m_2,x_1^N,x_2^N,y^N)
\nonumber \\
& \leq & \min  \left\{ \sum_{m_1'\neq m} P(A_{m_1'}|m_1,m_2,x_1^N,x_2^N,y^N), 1 \right\}\nonumber\\
& \leq & \left[ \sum_{m_1'\neq m_1}
P(A_{m_1'}|m_1,m_2,x_1^N,x_2^N,y^N) \right] ^\rho;
\qquad \text{any } 0 \leq \rho \leq 1 \nonumber\\
& \leq & \left[ (M_1-1) \sum_{{x'}_1^N} Q ({x'}_1^N||z^{N-1})
\left[\frac{
P(y^N||{x'}_1^N,x_2^N)}{P(y^N||{x}_1^N,x_2^N)}\right]^s\right]
^\rho, \qquad 0 \leq \rho \leq 1, s>0,\nonumber \\
\end{eqnarray}
where the last inequality is due to inequality (\ref{e_pAm}). By
substituting inequality (\ref{e_MLB2}) in eq. (\ref{e_pem}) we
obtain:
\begin{eqnarray}
{\bf E}[P_{e1}] \leq
(M-1)^{\rho}\sum_{y^N,x_2^N}Q(x_2^N||z^{N-1})\left[ \sum_{x^N}
Q(x_1^N||z_1^{N-1})P(y^N||x_1^N,x_2^N)^{1-s\rho}\right]\left[
\sum_{{x'}_1^N} Q
({x'}_2^N||z^{N-1})P(y^N||{x'}_1^N,x_2^N)^{s}\right]^\rho\nonumber
\end{eqnarray}
By substituting $s=1/(1+\rho)$, and recognizing that $x'$ is a
dummy variable of summation, we obtain eq. (\ref{e_MLB_a}) and
complete the proof of the bound on ${\bf E}[P_{e1}]$.

The proof for bounding  ${\bf E}[P_{e2}]$ is identical to the
proof that is given here for ${\bf E}[P_{e1}]$, up to exchanging
the indices. For ${\bf E}[P_{e3}]$ the upper bound is identical to
the case of the point-to-point channel with an input
$x_1^N,x_2^N$, as proven in \cite{Permuter06_feedback_submit}
where the union bound which appears here in eq. (\ref{e_MLB2})
consists of $(M_1-1)(M_2-1)$ terms.
\hfill \QED

\bibliographystyle{IEEEtran}
\bibliography{fs_mac_feedback_transaction}

\end{document}